%% file: mmWaveCoverageWithMIMOandMU.tex
\documentclass[twocolumn,10pt]{IEEEtran}

\input{def.tex}

\usepackage{dsfont}
\usepackage{minibox}
\DeclareSymbolFont{matha}{OML}{txmi}{m}{it}
\DeclareMathSymbol{\varv}{\mathord}{matha}{118}
\usepackage{eurosym}
\usepackage{hhline}
\usepackage{multicol}
\usepackage{eurosym}

\makeatletter

\IEEEoverridecommandlockouts
\begin{document}
\title{Nuts and Bolts of a Realistic Stochastic Geometric Analysis of mmWave HetNets: Hardware Impairments and Channel Aging}
\author{Anastasios Papazafeiropoulos, Tharmalingam Ratnarajah, Pandelis Kourtessis, and Symeon Chatzinotas   \vspace{2mm} \\
\thanks{Copyright (c) 2015 IEEE. Personal use of this material is permitted. However, permission to use this material for any other purposes must be obtained from the IEEE by sending a request to pubs-permissions@ieee.org.}
\thanks{A. Papazafeiropoulos was with the Institute for Digital  Communications, University of Edinburgh, Edinburgh EH9 3JL, U.K. He is  
now with the Communications and Intelligent Systems Group, University of Hertfordshire, Hatfield AL10 9AB, U.K., and also with SnT,  
University of Luxembourg, L-1855 Luxembourg, Luxembourg. T. Ratnarajah is  with the  Institute for Digital Communications (IDCOM), University of Edinburgh, Edinburgh  EH9 3JL,  U.K. P. Kourtessis is with the Optical Networks Research Group,
University of Hertfordshire, Hatfield AL10 9AB, U. K.  S. Chatzinotas is with SnT,  
University of Luxembourg, L-1855 Luxembourg, Luxembourg. 
E-mails:  tapapazaf@gmail.com, t.ratnarajah@ed.ac.uk, p.kourtessis@herts.ac.uk, symeon.chatzinotas@uni.lu.}
\thanks{This work was supported in part by the U.K. Engineering and Physical  
Sciences Research Council (EPSRC) under Grant EP/N014073/1, and in part by UK-India Education and Research Initiative Thematic  
Partnerships under Grant DST-UKIERI-2016-17-0060 and by FNR, Luxembourg, through the CORE project ECLECTIC.}}\maketitle
 \vspace{-1cm}
\begin{abstract}
Motivated by  heterogeneous networks (HetNets) design in improving coverage and by millimeter-wave (mmWave) transmission  offering an abundance of extra spectrum, we present a general analytical framework shedding light to the downlink of realistic mmWave HetNets consisting of $K$ tiers of randomly located base stations (BSs). Specifically, we model, by virtue of stochastic geometry tools, the \textit{multi-tier} \textit{multi-user multiple-input multiple-output (MU-MIMO) mmWave} network degraded by the inevitable \textit{residual additive  transceiver hardware impairments (RATHIs)} and \textit{channel aging}. Given this setting, we derive the coverage probability and the area spectral efficiency (ASE), and we subsequently  evaluate the impact of residual transceiver hardware impairments (RTHIs) and channel aging on these metrics. Different path-loss laws for line-of-sight (LOS) and non-line-of-sight (NLOS) are accounted for the analysis, which are among the distinguishing features of mmWave systems. Among the findings, we show that the RATHIs have a meaningful impact at the high signal-to-noise ratio (SNR) regime, while the transmit additive distortion degrades further than the receive distortion the system performance. Moreover, serving fewer users proves to be preferable, and the more directive the mmWaves are, the higher the ASE becomes.
\end{abstract}
\begin{keywords}
Heterogeneous MIMO networks, millimeter wave transmission systems, channel aging, transceiver hardware impairments, performance analysis.
\end{keywords}

\section{Introduction}
Current cellular networks have manifested an exponential increase in traffic load~\cite{Index2013}. In this direction, among the dominant effective ways to increase the network capacity in the forthcoming fifth generation (5G) networks is the cell densification, which reduces user distance  since the base station (BS) density becomes very large~\cite{Osseiran2014}. In fact, the advent of hotspots has improved the coverage and  spatial reuse,  has achieved efficient offloading of the traffic,  and has boosted the spectral efficiency per unit area. In order to avoid Monte-Carlo simulations, tractable and accurate models have been  introduced via the theory of Poisson point processes (PPPs) to describe the randomness concerning the locations of the BSs. Having started from the downlink of single-input single-output (SISO) systems~\cite{Andrews2011}, an extension to multi-tier network has been encountered in~\cite{Dhillon2012}. In a similar way, the coexistence of MIMO and stochastic geometry in  heterogeneous networks (HetNets) has been actualized in~\cite{Dhillon2013}.

In a parallel avenue, another key technology, aiming to achieve the increased capacity demand, is millimeter-wave (mmWave) transmission systems. Such systems offer large portions of the unused spectrum, which can be exploited for boosting the data rate~\cite{Rappaport2014,Pi2011,Rappaport2013}. Commercial wireless systems such as the IEEE 802.11ad for local area networking~\cite{Alejos2008} have already considered the mmWave band ranging from 30 GHz to 300 GHz, where field measurements have taken place recently~\cite{Pi2011,Alejos2008}. 

Notably, mmWave cellular systems employ large antenna arrays to  benefit from the application of beamforming that can compensate the frequency-dependent path-loss, and reduce the out-of-cell interference~\cite{Rappaport2014}.  At the same time, the smaller the  wavelength, the smaller the antenna aperture to be used. Hence, multiple antenna elements can be packed into a smaller volume, i.e., mmWave transmission is required to implement hundreds or thousands of antennas to a practical cost-efficient BS.   Unfortunately,  mmWave communication systems become more sensitive to blockage effects with increasing frequency~\cite{Alejos2008}. Different materials have different penetration losses~\cite{Alejos2008}. Interestingly, channel measurements with directional antennas have revealed that blockages cause substantial differences in the line-of-sight (LOS) paths and non-line-of-sight (NLOS) path-loss features~\cite{Pi2011,Rappaport2013a,Rajagopal2012}. In other words, in the case of a comprehensive system analysis for mmWave systems, the characteristics of the propagation environment should be accurately accounted for because the use of directional beamforming does not allow the application of results from the conventional analysis. 

The application of stochastic geometry to study mmWave cellular networks is limited to  a few substantial works~\cite{Akoum2012,DiRenzo2015,Turgut2017,Rebato2017}. For example, reference~\cite{Akoum2012} presents a basic work, where directional beamforming is considered for single and multiple users describing a simplified path-loss model  without considering any mmWave propagation characteristics. Similarly, in~\cite{DiRenzo2015}, the author assumed that the actual array beam pattern follows a step function with a constant main lobe over the beamwidth and a constant side lobe otherwise. Moreover, generalized fading (Nakagami-\textit{m}) was considered as a suitable distribution to model the LOS and NLOS components. Another example is~\cite{Turgut2017}, where  authors exploited the directional beamforming in mmWave cellular networks to improve the coverage probability by increasing the main lobe gain.






Remarkably, HetNets massive MIMO, and mmWave transmission  are currently mostly studied in isolation~\cite{Osseiran2014},  however, it is expected they will coexist in 5G systems and beyond to address some of the critical challenges. Although these solutions target the implementation of 5G networks, a number of technical misconceptions and challenges are met and remain unsolved. A fundamental example is the impact of the detrimental unavoidable residual transceiver hardware impairments (RTHIs) is a highly active area of great industrial interest, which has not been taken into account in next-generation networks. Unfortunately, although 5G networks, and especially, mmWave systems and HetNets should be the cynosure regarding the study of residual transceiver hardware impairments (RTHIs), the situation is different as a literature survey reveals. In the majority of  literature,  ideal transceiver hardware is assumed which is far form realistic, considering the inevitability of RTHIs~\cite{Schenk2008,Studer2010,Goransson2008,Bjoernson2013,PapazafeiropoulosMay2016,PapazafeiropoulosJuly2016,PapazafComLetter2016,PapazafeiropoulosJuly2016,Papazafeiropoulos2017a,Papazafeiropoulos2017,Papazafeiropoulos2017c,Bjornson2015,Papazafeiropoulos2016}\footnote{The RTHIs denote the amount of distortions, which occur by the partially mitigation of the transceiver hardware because although real-world applications employ calibration schemes at the transmitter and  compensation algorithms at the receiver, the transceiver hardware impairments are only partially mitigated~\cite{Schenk2008,Bjornson2015,Papazafeiropoulos2016}}.  In particular, the additive RTHIs (RATHIs), which describe the aggregate effect of many impairments,  are modeled as additive Gaussian noises at both the BS and user's side~\cite{Schenk2008,Studer2010}. Note that the Gaussian model is adopted because of its analytical tractability and experimental validation~\cite{Schenk2008}.

In this context, the relative movement between  antennas and  scatterers, which is common in practical systems, results in  channel variation between what is learned via estimation and what is used for precoding or detection~\cite{Papazafeiropoulos2015a,Papazafeiropoulos2016}. This effect is known as channel aging, and its study appears  a gap in the 5G literature. Notably, the lack of study  of channel aging becomes more significant in outdoor urban environments that are characterized by  increased mobility. Especially, this work, including the concept of channel aging during mmWave transmission is quite meaningful since mmWave systems are very sensitive to outdoor communications with high velocities~\cite{Rappaport2013}.  Hence, given that user mobility is one of the main causes for the inevitably imperfect channel state information at the transmitter (CSIT), it should be taken seriously into account.

\subsection{Motivation}
Most  existing works consider perfect hardware and CSIT, which are highly unrealistic assumptions. Thus, the inconsistency between theory and reality grows and results in misleading conclusions. Hence, this work relies on the recognition that the 5G solutions should consider the RTHIs and user mobility. In addition, in order to avoid Monte-Carlo simulations, we have introduced tractable and accurate models for HetNets  in terms of the theory of PPPs to describe the randomness of BSs locations. Also, the RTHIs have been taken into account only in~\cite{PapazafComLetter2016} and~\cite{Papazafeiropoulos2017}. Specifically,~\cite{PapazafComLetter2016} studies RTHIs in the case of perfect CSIT, while~\cite{Papazafeiropoulos2017} investigates the system under the presence of RTHIs, pilot contamination, and channel aging, which is closer to our work but not at the full extend of technologies we propose.  Furthermore, although in~\cite{Bai2015}, the authors have taken into consideration the mmWave condition,  their scenario focuses only on unrealistic assumptions of perfect CSIT and hardware. As a result, noticing the marriage of many studies between HetNets and mmWave transmission, we enrich the general setting of HetNets with the special characteristics of high frequencies. In fact, we formulate a general practical MU-MIMO  with randomly-located BSs serving in the mmWave band, and impaired by the unavoidable RTHIs and imperfect CSIT due to channel aging. We focus on the determination of the potentials of HetNets enriched by the mmWave technology before their final implementation, in order to comply with the increasing need for conducting realistic characterization of 5G networks.

\subsection{Contributions}
The main contributions are summarised as follows.
\begin{itemize}
 \item We shed light on the impact of RTHIs and channel aging on the performance of the downlink coverage probability and the ASE of MU multiple-antenna BSs employing mmWave transmission in a HetNet design. In our investigations, we take into account for the residual additive distortions at both the transmitter and the receiver as well as amplified thermal noise (ATN) in a general realistic scenario, where only imperfect CSIT is available. For the sake of comparison, we also present the results corresponding to perfect hardware.
 \item Contrary to existing works~\cite{PapazafComLetter2016,Papazafeiropoulos2017}, which have studied the effect of RTHIs on the performance of heterogeneous cellular networks with perfect and imperfect CSIT, respectively, we focus on the mmWave band with its special characteristics. Moreover, with comparison to~\cite{Bai2015}, we introduce the RTHIs and channel aging, as well as we, investigate a more general setting which includes a multi-tier multi-user (MU)  mmWave setup with multiple BS antennas instead of single-tier single-input single-output (SISO) channels.
 \item We present several observations that proper system design should take into account. We show how the additive distortions and the amplified thermal noise degrade system performance. Furthermore, we quantify the degradation of the system due to time variation of the channel.
 \item The drawn characteristics confirm that  the higher the channel aging becomes, i.e., with increasing user mobility, the more severe the degradation of the systems is, especially for certain Doppler shifts. In  addition,
 the additive  distortion at the transmitter
 has a higher impact than the distortion at the receiver side, but, in general, the impact of both RATHIs becomes apparent at high signal-to-noise ratio (SNR). At the same time, the ATN contributes to the degradation of the system at low SNR, while at high SNR, ATN is negligible. It is also apparent that by increasing the directivity of the main lobe of the mmWave transmission,  ASE   increases resulting in a tradeoff. 
In fact, a trade-off between quality and cost should be chosen. 
\item We have also provided information on how the number of BS antennas and users can affect the drawn results. It is shown that it is better to employ more antennas at the BS, which agrees with the general idea of combining massive MIMO and mmWave transmission in a HetNet design. Also, it is a better design choice for every BS to serve as few users as possible.
  \end{itemize}
\subsection{Paper Outline}  
The remainder of this paper is organised as follows.  Section~\ref{System} develops the system model of a realistic multi-tier MU-multiple-input multiple-output (MIMO) mmWave HetNet with channel aging and RTHIs operating at mmWave frequencies.  In the same section,  channel aging is introduced and a description of the RATHIs is provided.  Next, Section~\ref{downlink} presents the downlink mmWave transmission under RATHIs and imperfect CSIT. Section~\ref{main} provides the main results of this work. Especially, Subsection~\ref{coverage} includes the derivation and investigation of the coverage probability, while Subsection~\ref{AverageAchievableRate1}, provides the presentation of the  ASE under the same realistic conditions.  The numerical results are placed in Section~\ref{Numerical}, and Section~\ref{Conclusion} concludes the paper.
\subsection{Notation}Vectors and matrices are denoted by boldface lower and upper case symbols. The symbols $(\cdot)^\T$, $(\cdot)^\H$, and $\tr\!\left( {\cdot} \right)$ express the transpose,  Hermitian  transpose, and trace operators, respectively. The expectation  operator is denoted by $\EE\left[\cdot\right]$, while the $\mathrm{diag}\{\cdot\}$ operator generates a diagonal matrix from a given vector, and the symbol $\triangleq$ declares definition. The notations $\mathcal{C}^{M \times 1}$ and $\mathcal{C}^{M\times N}$ refer to complex $M$-dimensional vectors and  $M\times N$ matrices, respectively. The indicator function $ \mathds{1}(e)$ is $1$ when event $e$ holds and $0$ otherwise. Moreover, $\mathrm{J}_{0}(\cdot)$ is the zeroth-order Bessel function of the first kind, and $\Gamma\left( x,y \right)$ denotes the Gamma distribution with shape and scale parameters $x$ and $y$, respectively. Furthermore, $\underset{x \in A}{\cup}$ denotes the union with $A$ being an index set. Also, $\mathcal{L}_{I}\!\left(s \right)$ expresses the Laplace transform of $I$. Finally, $\bb \sim \cC\cN{(\b0,\mathbf{\Sigma})}$ represents a circularly symmetric complex Gaussian vector with zero mean and covariance matrix $\mathbf{\Sigma}$.

\section{System Model}\label{System}
This section introduces the downlink model of a realistic HetNet embodying the principles of PPP modeling as well as MU-MIMO and mmWave transmission under the presence of imperfect CSIT with channel aging and inevitable hardware impairments.
\subsection{General Characteristics}
We consider  a set of $W$ different classes (tiers) of BSs with $\mathcal{W}=\{1, 2, \ldots,W\}$, where hundreds of femtocells coexist in each
macrocell with a multi-antenna BS and multiple single-antenna users per cell\footnote{In practice, such a setting means that the BSs of femtocells transmit with orders of magnitude lower power than macrocells, have a smaller number of antennas, and serve a smaller number of users.}. In addition, capturing the deployment trends in 5G networks (massive MIMO), each BS can employ a  number of antennas $N_{w}$, which can be quite large by approaching the regime of massive MIMO systems~\cite{Larsson2014}, which, in turn, is suggested by mmWave technology~\cite{Rappaport2013}. In such a case, many degrees of freedom are available to share per cell. The locations of the BSs of each tier are drawn from a general stationary point process $\Phi_{\mathrm{B}_{w}}$ with deployment density $\lambda_{\mathrm{B}_{w}}$. More compactly, we imply an MU-MIMO HetNet formulation, where $K_{w}\le N_{w}$ users, that are independently distributed with a comparison to the BSs and blockages on the plane, belong to the $w$th Voronoi cell. Across tiers, the BSs differ in terms of the transmit power $\rho_{w}$, the number of antennas $N_{w}$,  the number of users  $K_{w}$ served by each BS in a given resource block, and target signal-to-distortion-plus-interference-plus-noise ratios (SDINRs) $T_{w}$. In essence, each macro BS serves a higher load than its femto counterpart.  Similarly, the user locations in the $w$th tier are  modeled by means of a stationary independent PPP $\Phi_{\mathrm{K}_{w}}$ with  a sufficiently high density $\lambda_{\mathrm{K}_{w}}$ such that $K_{w}$ users are associated per BS\footnote{It is worthwhile to mention, that in reality, the various parameters also differ across the tiers. For example, such parameters are the number of antennas per BS, and the number of associated users. Hence, the BSs  in the $w$th tier include $N_{w}$ antennas and serve $K_{w}$ users. For the sake of simplicity, we indicate that the various parameters do not vary among the cells of tier $w$, i.e., the number of BS antennas in the $w$ tier is   $N_{w}~\forall l$, where $l$ denotes the cell number.}. In other words, we assume that the size of each cell is so large that it can accommodate $K_{w}$ users. 

Blockages such as buildings comprise a stationary
and isotropic (invariant to the motions of translation and rotation) process of random shapes~\cite[Ch. 10]{Baccelli2010}.  The BSs can be arranged inside or outside the blockages. Focusing on the outside BSs of the $w$th tier, we denote $\tilde{\Phi}_{w}=\{X_{wl}\}$ their point process, while $X_{wl}$ represents the $l$th outdoor BS at tier $w$, and $R_{wl}$ represents  the distance between the $l$th BS in the $w$th tier and the origin $0$.  The average fraction of the indoor area in the network in the $w$th tier, defined by $\gamma_{w}$ coincides with the average fraction of the land covered by blockages. On this ground, the probability each BS to be located outdoor is i.i.d. and given by $1-\gamma_{w}$. Taking into account the thinning theorem of PPP~\cite{Baccelli2010}, the density of the  outdoor BS process $\tilde{\Phi}_{w}$, being a PPP, is  $\lambda_{w}=\left( 1-\gamma_{w} \right)\lambda_{\mathrm{B}_{w}}$.

The focal point of this work is the downlink transmission, initiated at a BS located outdoor and ending at a single-antenna user located at the origin, which is found outdoor\footnote{We consider that the outdoor user  cannot receive any signal or interference from an indoor BS  because we assume that the indoor-to-outdoor penetration loss is high enough in the mmWave band. Furthermore,  the coverage of the indoor users can be achieved by
either indoor BSs or by outdoor BSs operating at ultra high frequencies (UHFs) since they have smaller indoor-to-outdoor penetration losses.  }.  Exploiting Slivnyak's theorem, we are able to conduct the analysis by focusing on a typical user, being  a user chosen at random from amongst all users in the network~\cite{Chiu2013a}. Without loss
of generality, we assume that the typical user is located at the origin. We assume that the serving  BS is located at $X_{0}$. The user with the smallest path-loss $L\left( R_{w}\right)$ associates with this BS.  We neglect the index $w$ since we refer to the current tier ($w$th tier). The coverage region of each outdoor BS defines the region with the maximum average received power, and the set of all cells constitutes a weighted Voronoi tessellation. 

\begin{table*}[!t]
\begin{center}\caption {Probability Mass Functions OF $G_{wl}$ \cite{Bai2015}} \label{tablaa}
\begin{tabular}{ |c|| c| c| c| c|}
 \hline
 $k$ & $1$ & $2$& $3$&$4$ \\
\hhline{|=#=|=|=|=|}
  $a_{wk}$ & $M_{w\mathrm{r}}M_{w\mathrm{t}}$ & $M_{w\mathrm{r}}m_{w\mathrm{t}}$&
     $m_{w\mathrm{r}}M_{w\mathrm{t}}$&$m_{w\mathrm{r}}m_{w\mathrm{t}}$ \\  
   \hline
$b_{wk}$ & $c_{wr}c_{wt}$ & $c_{wr}\left( 1-c_{wt} \right)$& $\left( 1-c_{wr} \right)c_{wt}$&$\left( 1-c_{wr} \right)\left( 1-c_{wt} \right)$\\
 \hline
\end{tabular}
\end{center}
\hrulefill
\end{table*}
Notably, we invoke that a BS is LOS to the typical user, found at the origin, when there is no blockage between them. Reasonably, in each tier, the blockages allow  an assortment of outdoor BSs to be LOS, while the rest BSs are NLOS. In other words, an outdoor BS can be discerned to NLOS and LOS to the typical user. We denote  $\Phi_{\mathrm{L}_{w}}$ the point process of LOS BSs, while $\Phi_{\mathrm{N}_{w}}=\Phi_{\mathrm{B}_{w}}\backslash \Phi_{\mathrm{L}_{w}}$ is the  process of NLOS BSs. $N_{w\mathrm{L}}$ and $N_{w\mathrm{N}}$ express the number of LOS and NLOS BSs. Moreover, the probability that a link of length $R_{wl}$ is LOS is called LOS probability function and is denoted by $p\left( R_{wl}\right)$. In particular,  the LOS probability function depends only on the length of the link $R_{wl}$ because the distribution of the blockage process has been assumed stationary and isotropic. Also, $p\left( R_{wl}\right)$ is a non-increasing function of $R_{wl}$ since the shorter the link, the more unlikely it will be intersected by one or more BSs. Obviously, the NLOS probability
of a link is $1-p\left( R_{wl}\right)$. The LOS probability function in a network can be obtained by means of stochastic blockage models~\cite{Bai2014} or field measurements~\cite{Akdeniz2014}, while the blockage parameters can be defined by some random distributions. For example, we have that $p\left( R_{wl}\right)=\mathrm{e}^{-\beta_{w} R_{wl}}$ with  $\beta_{w}$ being a parameter described by the
density and the average size of the blockages ($1/\beta_{w}$ is called the average LOS range of the network), if the blockages are modeled as a rectangle Boolean scheme~\cite{Bai2014}. Moreover, in this work, we ignore any correlations of blockage effects between the links, and, as a result, the  LOS probabilities  are assumed to be independent. Also, the LOS  and the NLOS BS processes are assumed independent with density functions $p\left( R_{wl}\right)\lambda_{w}$ and $\left( 1-p\left( R_{wl}\right) \right)\lambda_{w}$, respectively. In addition, the LOS and NLOS links obey to different path loss laws. Hence, the path-loss $L\left( R_{wl}\right)$, where $R_{wl}$ is the length of the link in polar coordinates, is obtained by
\begin{align}
 L\left( R_{wl}\right)&=\mathbb{I}\left( p\left( R_{wl}\right) \right)C_{\mathrm{L_{w}}}R_{wl}^{-\al_{\mathrm{L_{w}}}}\nn\\
 &+\left( 1-\mathbb{I}\left( p\left( R_{wl}\right) \right)\right)C_{\mathrm{N}_{w}}R_{wl}^{-\al_{\mathrm{N}_{w}} },
\end{align}
where $\mathbb{I}\left( x \right)$ denotes a Bernoulli random variable with parameter $x$, while $\al_{\mathrm{L_{w}}}$, $C_{\mathrm{L_{w}}}$ and $\al_{\mathrm{N}_{w}}$, $C_{\mathrm{N}_{w}}$ are the LOS and   NLOS path loss exponents,  intercepts of the LOS and NLOS BSs, respectively. Prior works such as~\cite{Rappaport2013} provide typical values for the mmWave path loss exponents and intercept constants.

 {Directional beamforming by means of antenna arrays deployed at the BSs is another assumption that could also hold for the mobile stations. However, for the sake of exposition, we assume single-antenna users as already stated. In order to make the analysis tractable, the  array patterns from the $l$th BS $G_{M_{w\mathrm{i}},m_{w\mathrm{i}},\theta_{wl\mathrm{i}}}\left( \phi_{wl\mathrm{i}} \right)$, where  $M_{w\mathrm{i}}$ is the main lobe directivity gain, $m_{w\mathrm{i}}$ is the back lobe gain, $\theta_{wl\mathrm{i}}$ is the beamwidth of the main lobe, and $\phi_{wl\mathrm{i}}$ is the angle of the boresight direction, are approximated by a sectored antenna model as in~\cite{Bai2015}. Note that the index $i$ takes two values. If it is $\mathrm{t}$, it describes the parameters of the BS, while if it is  $\mathrm{r}$, it  represents the variables of the user (mobile station). Setting the boresight direction of the antennas equal to $0^{\mathrm{o}}$, the total directivity gain in the
link between the $l$th BS and the typical user is  $G_{wl}=G_{M_{w\mathrm{t}},m_{w\mathrm{t}},\theta_{wl\mathrm{t}}}\left( \phi_{wl\mathrm{t}} \right)G_{M_{w\mathrm{r}},m_{w\mathrm{r}},\theta_{wl\mathrm{r}}}\left( \phi_{wl\mathrm{r}} \right)$ with $\phi_{wl\mathrm{t}}$ and $\phi_{wl\mathrm{r}}$ being the angle of departure and the angle of arrival  of the signal. The directivity gain in an interference link $G_{wl}$ is a discrete random variable with  probability distribution   $G_{wl}=\al_{wk}$  with probability $b_{wk}$, where $k=1,2,3,4$. The constants $\al_{wk}$ and $b_{wk}$ are given in Table~\ref{tablaa}, where $c_{w\mathrm{r}}=\frac{\theta_{w\mathrm{r}}}{2\pi}$
as $c_{wl\mathrm{t}}=\frac{\theta_{w\mathrm{t}}}{2\pi}$. The random directivity gain $G_{wl}$ for the $l$th interfering link results, if we assume that the angles  $\phi_{wl\mathrm{t}}$ and  $\phi_{wl\mathrm{r}}$
are assumed to be independently
and uniformly distributed in $(0, 2 \pi]$. Especially, in the case of the directivity
gain for the desired signal link, it is $G_{w} =M_{w\mathrm{r}}M_{w\mathrm{t}}$\footnote{ Given that the proposed model and analysis are quite flexible, their extensive performance under 3GPP practical models is an interesting topic for futute work due to limited space. As an example, we provide an 3GPP antenna gain pattern $G_{\mathrm{3GPP}}\left( \theta \right)$ defined by
\begin{align}
G_{\mathrm{3GPP}}\left( \theta \right)=\left\{\begin{array}{ll}g_{1}10^{-\frac{3}{10}\left( \frac{2 |\theta|}{|\omega|} \right)^{2}}&~\mathrm{if}~|\theta|\le \theta_{1}\\                                                                                                                                                                                                                                                                                                                                                                                                                                                                                                                                                                                                                                                                                                                                                                                                                                                                                                                                                                                                                                                             g_{2}&~\mathrm{if}~\theta_{1}<|\theta|\le \pi,   \\                                                                                                                                                                                                                                                                                                                                                                                                                                                                                                                                                                                                                                                                                                                                                                                                                                                                                                                                                                                                                                                            \end{array}
\right. \end{align}
where $\omega$ is the 3 dB-beamwidth, $\theta_{1} = \omega/2 \sqrt{10/3 \log_{10}(g_{1}/g_{2})}$, while  $g_{1}$ and $g_{2}$ are the max
 and side-lobe gains, respectively with $0 \le g_{2} < g_{1}$~\cite{UMTS}.}.}

\subsection{CSIT Model}
In practical systems, the CSIT available at the transmitter can be imperfect due to several reasons. In this work, we focus on the  lack of accuracy due to limited feedback and channel aging. Below, we describe these two sources, which are present in both small and large antenna regimes~\cite{Bjoernson2015,Papazafeiropoulos2015a,Papazafeiropoulos2017}, and are quite meaningful in mmWave systems. As a result, the proposed model is capable of describing any number of antennas. As far as the small-scale fading is concerned, we assume independent Rayleigh fading with different parameters for each link.
\subsubsection{Channel Estimation}
During the system design, a selection between time division duplex (TDD) and frequency division duplex (FDD) is made according to the requirements and the constraints~\cite{Bjoernson2015}. Although massive MIMO systems, suffering from pilot contamination, employ the former design~\cite{Rusek2013}, the implementation of an FDD solution due to existing infrastructure is viable and can be employed in both small and large number of antennas BS designs~\cite{Bjoernson2015}. Herein, for the sake of simplicity and without loss of any generality, we assume FDD, where the BS has available imperfect CSIT due to limited feedback, e.g., a quantized feedback with a fixed number of quantization bits~\cite{Wagner2012,Bjoernson2015}. Thus, the estimated channel at the associated BS of the $w$th tier is given by 
\begin{align}
 {\bh}_{w}=\sqrt{1-\tau_{w}^{2}}\hat{\bh}_{w}+\tau_{w}\tilde{\bh}_{w},
\end{align}
where $\tilde{\bh}_{w}$, being the estimation error,  has i.i.d. $\mathcal{CN}{\left( 0,1 \right)}$ entries independent of $\hat{\bh}_{w}$. Note that $\tau_{w}\in\left[ 0,1\right]$ is a parameter indicating the quality of instantaneous CSIT for the associated BS. For example, $\tau_{w}=0$  denotes perfect CSIT, whereas $\tau_{w}=1$ expresses that the estimated CSIT and perfect channel are completely uncorrelated. Note that both network nodes, i.e., the BS and the user, calculate their channels (angles of arrivals and fading) driven to profit the maximum directivity gain by adjusting their antenna steering orientations. 
\subsubsection{Channel Aging}
In common environments, relative mobility of the users with a comparison to the BS antennas takes place. Hence, the channel varies with time, and the result is a time-varying CSIT model~\cite{Papazafeiropoulos2015a}. Mathematically, we consider an autoregressive model of order $1$, where the current sample is related to its previous sample, that depends on the second-order statistics of the channel in terms of its autocorrelation function. Note that the autocorrelation function is generally a function of the velocity of the user, the propagation geometry, and the antenna characteristics. More concretely, we  ponder a Gauss-Markov model of low order $\left( 1 \right)$ for reasons of computational complexity and tractability~\cite{Truong2013,Papazafeiropoulos2015a}.  In such case, the current channel at the $w$th tier between the associated BS  and the typical user belonging to its cell tier is related to its previous sample as
\begin{align}
\bh_{w,n}  =& \delta_{w} \bh_{w,n-1} + \bee_{w,n},\label{eq:GaussMarkoModel}
\end{align}
where $\bh_{w,n-1}$ is the channel in the previous symbol duration and $\bee_{w,n} \in \bbC^{N_{w}}$ is an uncorrelated channel error due to the channel variation modeled as a stationary Gaussian random process with i.i.d.~entries and distribution $\cC\cN(\b0,(1-\delta_{w}^2)\Id_{N_{w}}$~\cite{Vu2007}. Regarding $\delta_{w}$, it is related to the autocorrelation function, and it is provided by the following line of reasoning.  Specifically, we engage the Jakes model for the autocorrelation function, which is widely accepted due to its generality and simplicity \cite{Baddour2005}. Note that the Jakes model describes a propagation medium with two-dimensional isotropic scattering and a monopole antenna at the receiver \cite{WCJr1974}.   Mathematically, the normalized discrete-time autocorrelation function of the fading channel in the $w$th tier is expressed by
\begin{align}
r_{w}[k] 
=& \mathrm{J}_{0} (2 \pi f_{D_{w}}T_{s_{w}}|k|),\label{eq:scalarACF}
\end{align}
where  $|k|$, $f_{D_{w}}$, and $T_{s_{w}}$ are the delay in terms of the number of symbols, the maximum Doppler shift, and the channel sampling period, respectively. Concerning the maximum Doppler shift $f_{D}$, it can be expressed  in terms of the relative velocity of the  $v$, i.e., $f_{D_{w}}=\frac{v_{w} f_{c}}{c}$, where $c=3\times10^{8}~\nicefrac{m}{s}$ is the speed of light and $f_{c}$ is the carrier frequency. For the sake of simplicity, we assume $k=1$, and that the associated BS has perfect knowledge of $\delta_{w}=r_{w}[1]$.

Both effects of limited CSIT and time-variation of the channel can be combined. Specifically, the fading channel at time slot $n$ can be expressed by
\begin{align}   
 \bh_{w,n}&=\delta_{w} {\bh}_{w,n-1}+\bee_{w,n}\nonumber\\
&=\delta_{w} \sqrt{1-\tau_{w}^{2}} \hat{\bh}_{w,n-1}+ \tilde{\bee}_{w,n},\label{eq:MMSEchannelEstimate}
\end{align}
where $\hat{\bh}_{w,n-1}$ and $\tilde{\bee}_{w,n}= \delta_{w} \tau_{w} \tilde{\bh}_{w,n-1}+\bee_{w,n}\sim \mathcal{CN}\left(\b0,\sigma_{w,\tilde{\bee}_{w}}^{2} {\mathrm{ \bI}}_{M} \right)$ with $\sigma_{w,\hat{\bee}_{w}}^{2}=\left(1-\delta_{w}^{2}\left( 1-\tau_{w}^{2} \right) \right)$
are mutually independent. In other words, the estimated channel at time $n$ is now $\hat{\bh}_{w,n}=\delta_{w}  \sqrt{1-\tau_{w}^{2}} \hat{\bh}_{w,n-1}$. Given that, especially in highly mobile scenarios, misalignments such as imperfect antenna steering and suboptimal directivity gain may emerge, we have left their study as a topic of future research.  Note that beem steering concerns the change of the direction of the main lobe of a radiation pattern. The current model assumes perfect antenna steering and maximum directivity gain,  where, especially, perfect antenna steering means that the main lobes between each transmitter and receiver pair are aligned.

\subsection{Hardware Impairments}
The transceiver of practical systems includes unavoidable hardware imperfections. Herein, we examine the impact of RATHIs and ATN. The study of multiplicative impairments such as the phase noise is left for future work.
\subsubsection{Emergence of RATHIs}\label{Presentation} 
Despite the mitigation schemes, implemented in both the transmitter and receiver,  RATHIs still emerge by means of residual additive distortion noises~\cite{Schenk2008,Studer2010}. Hence,  the transmitter side introduces an impairment causing a mismatch between the intended signal and what is actually transmitted during the transmit processing, while at the receiver side the received signal appears a distortion.  

Especially, the majority of HetNets literature, except~\cite{PapazafComLetter2016,Papazafeiropoulos2017}, relies on the assumption of perfect transceiver hardware with the hardware imperfections being ignored. In this direction, the gap between theory and practice increases. Interestingly, steps forward towards a more realistic approach necessitate the incorporation of RATHIs in the design. In fact, from conventional wireless systems and continuing to 5G networks such as massive MIMO systems, the inclusion of RATHIs in the analysis results in more down-to-earth conclusions ~\cite{Schenk2008,Studer2010,Goransson2008,Qi2012,Bjornson2015,Papazafeiropoulos2016,Bjoernson2013,Papazafeiropoulos2017,Papazafeiropoulos2017a}. 

In mathematical terms, it has been shown by means of measurement results that the conditional transmitter and receiver distortion noises for the $i$th link, given the channel realizations, are modeled as Gaussian distributed having average power proportional to the average signal power~\cite{Studer2010}. The justification behind the circularly-symmetric complex Gaussianity relies on the aggregate contribution of many impairments. Moreover, since the additive distortions take new realizations for each new data signal, they are time-dependent. 

The RATHIs at the transmitter (associated BS) and the receiver (typical user) in the $w$th tier are given by
\begin{align}
 \etv_{w\mathrm{t},n}^{\mathrm{BS}}&\sim \cC\cN\left( \b0,\bm \Lambda^{\mathrm{BS}}_{w,n} \right)\label{eta_t} \\
 \eta_{w\mathrm{r},n}^\mathrm{UE}&\sim \cC\cN \left( \b0,\Upsilon^{\mathrm{UE}}_{w,n} \right)\label{eta_r},
\end{align}
where $\bm \Lambda^{\mathrm{BS}}_{w,n}= \kappa_{\mathrm{t}_\mathrm{BS}w,n}^{2}\mathrm{diag}\left( q_{1,n},\ldots,q_{M,n} \right)$ and 
$ \Upsilon^{\mathrm{BS}}_{w,n} =\kappa_{\mathrm{r}_{\mathrm{UE}}w,n}^{2}G_{w}L_{w}\left( R_{w} \right) \bh_{k,n}^{\H}\tr\left( \bQ_{BS,n} \right)\bh_{k,n} $, with $\bQ_{BS,n}$ being the transmit covariance matrix at time instance $n$ of the associated BS with diagonal elements $q_{\mathrm{i}_1,n},\ldots,q_{T_\mathrm{i},n}$. Hence, we have $\bm \Lambda^{\mathrm{BS}}_{w,n}= \kappa_{\mathrm{t}_\mathrm{BS}w,n}^{2}  \rho_{w}^{\mathrm{UE}}/M$ and $ \Upsilon^{\mathrm{UE}}_{w,n} =\kappa_{\mathrm{r}_{\mathrm{UE}}w,n}\sqrt{L_{w}\left( R_{w} \right)}M  \rho_{w}^{\mathrm{UE}}\|\bh_{k,n}^{\H}\|^{2}$. The proportionality parameters $\kappa_{\mathrm{t}_{BS}w,n}^{2}$ and $\kappa_{\mathrm{r}_{UE}w,n}^{2}$, where, in applications  are met as the error vector magnitudes (EVM) at each transceiver side, describe the severity of the residual impairments at the BS and the user~\cite{Holma2011}.  In particular, the requirements, concerning the proportionality parameteres provided by the long  term  evolution  (LTE) standard, are in the range $\left[ 0.08,0.175\right] $~\cite{ChannelEstimation}. Notably, practical mmWave-enabled massive MIMO systems, encouraged to be constructed by cheap equipments, will be characterized by larger values of $\kappa_{\mathrm{t}_{\mathrm{BS}}}$ and $\kappa_{\mathrm{t}_{\mathrm{UE}}}$, which are taken into consideration in this paper  as can be seen in Section~V.
\begin{remark}
The receive distortion at the typical user incorporates  the path-loss coming from the associated LOS or NLOS BS. 
\end{remark}
\subsubsection{ATN}\label{ATN} 
This impairment is modeled by the variance $\xi^{2}_{n}$ of a Gaussian distributed random variable with zero mean. In fact, it is expressed by an amplification of the thermal noise,  appearing as an increase of its variance~\cite{Bjornson2015}. In other words, we have $ \sigma^{2}\le \xi_{n} $, where $\sigma^{2}$ is the variance of the actual thermal noise. From the physical point of view, this amplification emerges from the low noise amplifier, the mixers at the receiver  as well as other components that  engender a relevant amplified effect.

\section{Downlink Transmission under Imperfect CSIT and RTHIs}\label{downlink} 
The purpose of this section is to model the downlink transmission and obtain the corresponding SDINR and the probability densities functions (PDFs) of its terms, in order to derive the coverage probability and the ASE. Based on the proposed MU-MIMO HetNet, employing mmWave transmission, the received signal at the typical user in the $w$th tier from its associated LOS/NLOS BS at $R_{w}$  during the transmission phase $n$ can be written as
\begin{align}
 y_{w,n}&=\sqrt{G_{w}L_{w}\left( R_{w}\right)}\bh_{w,n}^{\H}\left( \bs_{w,n}+\etv^{\mathrm{BS}}_{w\mathrm{t},n}  \right)+\eta^{\mathrm{UE}}_{w\mathrm{r},n}\nn\\
 &+\!\!\!\sum_{j\in \mathcal{W}}\sum_{l:X_{jl}\in \Phi_{B_{j}}\backslash X_{0}}\!\!\!\!\sqrt{G_{jl}L_{j}\left( R_{jl} \right)}\bg_{jl,n}^{\H}\bs_{jl,n}+z_{w,n}\label{signal},
\end{align}
where  $\bs_{w,n} =\bV_{w,n}  \bd_{w,n} \in \mathbb{C}^{N_{w} \times 1}$ is the  transmit signal vector from the associated LOS/NLOS BS at the $w$th tier with covariance matrix $\bQ^{\mathrm{BS}}_{w,n}=\EE\left[ \bs_{w,n}\bs^{\H}_{w,n}\right] =P_{w}^{\mathrm{BS}}\Id_{N_{w}}= \rho_{w}^{\mathrm{BS}}/N_{w}\Id_{N_{w}}$ and   $\tr\left( \bQ_{\mathrm{BS},n} \right)= \rho_{w}^{\mathrm{BS}}$ is the average transmit power. Also, $z_{w,n}$ is the amplified thermal noise. The channel vector $\bh_{w,n} \in \mathbb{C}^{N_{w}\times 1}$  represents the desired  channel vector between the associated BS located at $R_{w} $ at  time-instance $n$ and the typical user. In a similar manner,  $\bg_{jl,n} \in \mathbb{C}^{N_{j}\times 1}$ denotes the interference channel vector from the BSs found at $R_{jl}$ far from the typical user  at time-instance $n$.  Notably, in the special case of Rayleigh fading, the PDFs of the powers of both the direct and the interfering links follow the Gamma distribution~\cite{Dhillon2013,Papazafeiropoulos2017}.

Since the system setting includes a MU-MIMO design, for the sake of exposition, we employ  zero-forcing (ZF)  precoding to support multi-stream transmission. We denote $\bV_{w,n}=[\bv_{w1,n},\ldots,\bv_{wK_{w},n}]\in \mathbb{C}^{N_{w}\times K_{w}}$ the precoding matrix of the associated BS,  which multiplies the data signal vector $\bd_{w,n} = \big[d_{w1,n}
,\dots,~d_{wK_{w},n }
\big]^\T \in \mathbb{C}^{K_{w}}\sim \mathcal{CN}(\b0,P_{w}^{\mathrm{BS}}\bI_{K_{w}})$ for all users in that cell.
Especially, taking account~\eqref{eq:MMSEchannelEstimate}, the ZF precoder, engaged by the associated BS of the typical user, can be written as
\begin{align}\label{precoder}
{\bV}_{w,n}&=\bar{\bH}_{w,n}\left( \bar{\bH}^{\H}_{w,n}\bar{\bH}_{w,n} \right)^{-1}\\
&=\delta_{w}^{-1} \bar{\bH}^{\dagger}_{w,n-1}=\delta_{w}^{-1} \hat{\bV}_{w,n-1},  \label{delayedPrec}                                                                                                                                                                                                                                                                                                                                                                            \end{align}
where $\bar{\bH}_{w,n}$ is the normalised version of $\hat{\bH}_{w,n}$ given by $\bar{\bH}_{w,n}=\left[ \bar{\bh}_{w1,n},\ldots, \bar{\bh}_{K_{w},n}\right]\in \mathbb{C}^{\left( N_{w} \times K_{w} \right)}$ with columns $\bar{\bh}_{wi,n}=\frac{\hat{\bh}_{wi,n}}{\|\hat{\bh}_{wi,n}\|}$. Note that the average transmit power per user of the associated BS is constrained to $\rho_{w}^{\mathrm{BS}}$ since the precoder is normalised, i.e., $\mathbb{E}\Big[\mathrm{tr}\left( \hat{\bV}_{w,n}\hat{\bV}^{\H}_{w,n} \right)\Big]=1$. In \eqref{delayedPrec},  we have introduced the user mobility effect for the $k$th user by means of $\hat{\bh}_{w,n}=\delta_{w} \hat{\bh}_{w,n-1}$.
\begin{remark}
Interestingly,~\eqref{delayedPrec} illustrates the user mobility effect on the ZF precoder in the downlink  transmission between the associated BS and the typical user.
\end{remark}

\begin{assumption}
Given the increased path-loss and that mmWave transmission takes place, we assume that the RATHIs from other BSs are negligible.
\end{assumption}
\begin{assumption}
Although in HetNet studies the thermal noise is omitted due to very low impact, hereafter, based on simulations, we include the presence of thermal noise since its effect is not negligible with a comparison to the additive distortion noises and the interference coming from the other cells. 
\end{assumption}

Taking into consideration the practical conditions of a realistic transmission by including the assumptions of limited CSIT and channel aging (time variation of the channel, i.e., see~\eqref{eq:MMSEchannelEstimate}),  the downlink received signal by the typical user is given by
\begin{align}
 y_{w,n}&= \sqrt{G_{w}L_{w}\left( R_{w} \right)}\hat{\bh}_{w,n-1}^{\H}\hat{\bv}_{w,n-1} d_{wk,n} 
 \nn\\
 &+\delta_{w}^{-1}\sqrt{G_{w}L_{w}\left( R_{w} \right)}\tilde{\bee}_{w,n}^{\H}\hat{\bV}_{w,n-1} \bd_{w,n}\nn\\
&+\sqrt{G_{w} L\left( R_{w} \right)} \bh_{w,n}^{\H}\etv^{\mathrm{BS}}_{w\mathrm{t},n}+\eta^{\mathrm{UE}}_{w\mathrm{r},n}+z_{w,n}
\nn\\
 &+\!\sum_{j\in \mathcal{W}}\sum_{l:X_{jl}\in \Phi_{B_{j}}\backslash X_{0}}\!\sqrt{G_{jl} L\left( R_{jl} \right)}{\bg}_{jl,n}^{\H}\hat{{\bV}}_{jl,n }\bd_{jl,n}, \label{filteredsignal}
\end{align}
where we have substituted~\eqref{eq:MMSEchannelEstimate}, in order to replace the current desired channel vector with its estimated version. In addition, we have expressed the current precoder by means of its delayed instance, known at the associated BS, by considering~\eqref{delayedPrec}.  

\begin{remark}
If we assume a single-tier network operating in UHF  and TDD design, we result in~\cite{Papazafeiropoulos2017}. In addition, if we assume perfect CSIT, we obtain the model in~\cite{PapazafComLetter2016}.
Setting the additive distortion parameters in~\eqref{filteredsignal} to zero, neglecting the amplified thermal noise, assuming perfect CSIT and single-tier mode the closest signal model corresponding to the ideal downlink model with the mmWave transmission, which does not account for RATHIs is~\cite{Bai2015 }. 
Similar properties/observations hold for any other expression including the downlink RATHIs and the channel aging. 
\end{remark}

Encoding the message over many realizations of all sources of randomness in the model, described by~\eqref{filteredsignal},  we obtain the SDINR. Note that this model consists of the imperfect CSIT noise, the accompanied channel estimate error, and RATHIs. In order to facilitate the statistical description of the SDINR, we denote $Z_{w,n}$ the desired channel power from the associated BS at time $n$ located at $L_{w}\left( R_{w} \right)$ to the typical user, found at the origin. Similarly, we denote $I_{jl,n}$ the power of the interfering link from other BSs located at $L_{\mathrm{z}}\left( R_{jl} \right)$.

\begin{proposition}\label{SINR}
The  SDINR of the downlink transmission from the associated LOS/NLOS  BS in $\Phi_{\mathrm{z}_{w}}$ ($z \in \{\mathrm{L},\mathrm{N}\}$) to  the typical user at $R_{w} $, taking into account for   RATHIs and imperfect CSI due to limited feedback, and time variation of the channel due to user mobility,  is given by
 \begin{align}
   \!\!\!  \mathrm{SDINR}_{\mathrm{z}}\left(q_{w} , x_{w} \right)
    \!=\!
        \frac{
            \beta_{w{\mathrm{z}}} Z_{w,n} 
            }{\left( E_{w{\mathrm{z}},n}\!+\!
            I_{ \etv_{w\mathrm{t}{\mathrm{z}},n}}
            \!+\!I_{ \etv_{w\mathrm{r}{\mathrm{z}},n}}\right) \!+\!
     I_{{\mathrm{z}},n} +\xi^{2}_{w,n}},\label{SDIR1}
\end{align}
where $ \beta_{wz}=M_{w\mathrm{r}}M_{w\mathrm{t}} P_{w}^{\mathrm{BS}}C_{\mathrm{Z_{w}}}R_{w}^{-\al_{\mathrm{Z_{w}}}}$, and the PDF of the desired signal power $Z_{w,n}\!=\!|\hat{\bh}_{w,n-1}^{\H}\hat{\bv}_{w,n-1}|^2$, following a scaled Gamma distribution, is given by
\begin{align}\label{eq PDF1 1}
    p_{Z_{w,n}}\left( z \right)=
       \frac{
            e^{-z/\sigma_{\hat{\bh}_{w}}^{2}}
            }{
            \left(N_{w}-K_{w}\right)!
         \sigma_{\hat{\bh}_{w}}^{2}
            }
        \left(
            \frac{
                z
                }{
               \sigma_{\hat{\bh}_{w}}^{2}
                }
        \right)^{N_{w}-K_{w}}, ~ z \geq 0    
        \end{align}
        while the other power terms are provided by\footnote{We assume that the manufacturing characteristics of the LOS and NLOS BSs are the same, i.e., the additive impairments  do not change.}
\begin{align}
 &E_{w\mathrm{z},n}\!=\!\beta_{w{\mathrm{z}}}\delta_{w}^{-2} \left(1+\kappa_{\mathrm{t}_{\mathrm{BS}}w,n} ^{2}\right)
               \big\|
                    \tilde{\bee}_{w,n}^{\H}\hat{\bV}_{w,n-1}]
              \big\|^2\label{error}\\
           &I_{ \etv_{w\mathrm{t}\mathrm{z},n}}\!=\! \beta_{w{\mathrm{z}}}\kappa_{\mathrm{t}_{\mathrm{BS}}w,n} ^{2}\|{{\bh}}_{w,n}\|^2\label{transmit}\\
 &I_{ \etv_{w\mathrm{r}\mathrm{z},n}}\!=\!\beta_{w{\mathrm{z}}}{\kappa_{\mathrm{r}_{\mathrm{UE}}w,n}^{2}}\|{\bh}_{w,n}\|^{2}\label{receive}\\
&I_{{\mathrm{z}},n}\!=\! \sum_{j\in \mathcal{W}}\sum_{l:X_{jl}\in \Phi_{B_{j}}\backslash X_{0}}\!\!\!\!\!\!\!\!{G_{jl} L\left( R_{jl} \right)}P_{w}^{\mathrm{BS}}\|{\bg}_{jl,n}^{\H}\hat{{\bV}}_{jl,n }\|^{2}\nn\\
           &\!=\! \sum_{j\in \mathcal{W}}\sum_{l:X_{jl}\in \Phi_{B_{j}}\backslash X_{0}}\!\!\!\!\!\!\!\!{G_{jl} L\left( R_{jl} \right)}P_{w}^{\mathrm{BS}} g_{jl,n}.
            \end{align}
            \end{proposition}

            As can be seen in~\eqref{SDIR1}, the $\mathrm{SDINR}$ is a function of the position $x_{w}$ and $q_{w}$ defining a set of parameters. Specifically, we define $q_{w}\triangleq\{K_{w},N_{w},\lambda_{\mathrm{B}_{w}},\al_{w},\delta_{w},\tau_{w}, \kappa_{\mathrm{t}_\mathrm{BS}w},$ $ \kappa_{\mathrm{r}_\mathrm{UE}w},$ $\xi_{w}, M_{w\mathrm{t}},M_{w\mathrm{r}},m_{w\mathrm{t}},m_{w\mathrm{t}}\}$. 
\begin{proof}
See Appendix~\ref{SINRproof}.
\end{proof}
\begin{remark}
Each term of the denominator of~\eqref{SDIR1} describes different effects. Indeed, the terms from left to right indicate the estimation error, the transmit additive distortion noise, the receive distortion noise,  the inter-cell interference coming from other BSs belonging in different cells and tiers, and the last term expresses the ATN. Notably, the estimation error depends on both channel aging and additive transmit impairment. Note that the  term in the numerator expresses the desired signal contribution in the typical current cell.
 \end{remark}
\begin{remark}
The ideal mmWave model with no hardware impairments and channel aging is obtained if $\tau_{w}=\delta_{w}=1$,  $\kappa_{\mathrm{t}_{\mathrm{BS}}w,n}=\kappa_{\mathrm{r}_{\mathrm{UE}}w,n}=0$, and $\xi_{n} = \sigma^{2}$  $\forall w,~n$.
\end{remark}
\section{Main Results}\label{main}
This section presents the main results of this work  in terms of theorems, describing the coverage probability and the ASE of the typical user. Henceforth, we omit the time index $n$ for the sake of simplicity.

\subsection{Coverage Probability}\label{coverage} 
The investigation of the coverage probability, being one of the main cores of this work, is the topic of this section. Specifically, we derive an upper bound of the downlink coverage probability of the typical user in a multiple antenna HetNet operating at mmWave frequencies under the practical conditions of imperfect CSIT, channel aging,  and RATHIs. For this reason, it is important to provide first a formal definition of the coverage probability in the case of randomly located BSs. Next, we continue with the main result by means of a theorem, derived in Appendix~\ref{CoverageProbabilityproof}. Notably, despite the abstraction of the definition, we result in the most general expression known in the literature approaching a more realistic appraisal of a network with randomly located BSs operating at mmWave frequencies.

\begin{definition}[\!\!\cite{Dhillon2013,Papazafeiropoulos2017}]\label{def1}
A typical user  is  in  coverage if its effective  downlink SDINR from at least one of the randomly located BSs in the network is higher
  than the corresponding target $T_{w}$. In general, we have
 \begin{align}
 p_{c}\left(   q_{w}, T_{w}  \right) \triangleq\mathbb{P}\left(\underset{{w} \in \mathcal{W}}{ \bigcup} \max_{x_{w}\in \Phi_{\mathrm{B}_{w}}}\mathrm{SDINR\left( q_{w},x_{w} \right)>T_{w}}   \right).
\end{align}
 \end{definition}
 
Having provided the definition of the coverage probability, we present its expression by means of the following theorem.

\begin{theorem}\label{theoremCoverageProbability} 
The upper bound of the downlink coverage probability   $p_{c}\left(   T_{w},q_{w}  \right)$ of a general realistic MU-MIMO   HetNet  with randomly distributed multiple antenna BSs transmitting  in the mmWave region, accounting for RATHIs, imperfect CSIT, and channel aging is given by\footnote{$p_{c,\mathrm{L}}$ and $p_{c,\mathrm{N}}$ are obtained by means of Definition 1.}
\begin{align}
  {p_{c} \left(   q_{w}, T_{w}  \right) =
p_{c,\mathrm{L}}\left(   q_{w}, T_{w}  \right) +p_{c,\mathrm{N}}\left(   q_{w}, T_{w}  \right),}
\end{align}
where  $p_{c,\mathrm{L}}$ and $p_{c,\mathrm{N}}$ are the conditional coverage
probabilities given that the user is associated with a LOS in $\Phi_{\mathrm{L}_{w}}$ or a NLOS BS   in  $\Phi_{\mathrm{N}_{w}}$.

Specifically, 
$p_{c,\mathrm{z}}\left(   q_{w}, T_{w}  \right) $ with $z \in \{\mathrm{L},\mathrm{N}\}$  is given by
\begin{align}
 &p_{c,\mathrm{z}}\left(   q_{w}, T_{w}  \right)  =\!\underset{w \in \mathcal{W}}{\sum}A_{\mathrm{z}_{w}}\sum_{i=0}^{\Delta_{w}-1}\!\sum_{u=0}^{i}\!\sum_{u_1+u_2+u_3+u_4=i-u}e^{-s_{w\mathrm{z}}\xi^{2}}\!\!\!\nn\\
 &\times\int_{0}^{\infty}\!\!\!\binom{i}{u}\!\!
 \binom{i\!-\!u}{u_1+u_2+u_3+u_4}\frac{\left( -1 \right)^{i}\!s_{w\mathrm{z}}^{u}\left(  s_{w\mathrm{u}}\xi^{2} \right)^{u_{4}}  }{i!}\!\!\!\! \nn\\
&\times\frac{\mathrm{d}^{u_1}}{\mathrm{d}s_{w\mathrm{z}}^{u_1}}\mathcal{L}_{E_{w\mathrm{z}} }\!\left( s_{w\mathrm{z}} \right)\frac{\mathrm{d}^{u_2}}{\mathrm{d}s_{w\mathrm{z}}^{u_2}}\mathcal{L}_{I_{\etv_{t}w\mathrm{z}}}\!\!\left(s_{w\mathrm{z}} \right)\frac{\mathrm{d}^{u_3}}{\mathrm{d}s_{w\mathrm{z}}^{u_3}}\mathcal{L}_{I_{\etv_{r}w\mathrm{z}}}\!\!\left(s_{w\mathrm{z}} \right)\!\nn\\
&\times\frac{\mathrm{d}^{u}}{\mathrm{d}s_{w\mathrm{z}}^{u}} \mathcal{L}_{I_{\mathrm{z}}}\!\!\left(s_{w\mathrm{z}} \right)\hat{f}_{\mathrm{z}}\!\left( x_{w} \right)\mathrm{d}x_{w},
\end{align}
where $A_{\mathrm{z}_{w}}$   is defined in Lemma~\ref{lemma3}, and $s_{w\mathrm{z}}=\frac{{T}_{w}\beta_{w\mathrm{z}}^{-1}}{\sigma_{\hat{\bh}_{w}}^{2}}$.
$\mathcal{L}_{E_{w\mathrm{z}}}\!\left(s \right)$,  $\mathcal{L}_{I_{ \etv_{w\mathrm{t}z}}}\!\left(s \right)$,   $\mathcal{L}_{I_{ \etv_{w\mathrm{r}z}}}\!\left(s \right)$, and $\mathcal{L}_{I_{wl}}\!\left(s \right)$ are the Laplace transforms of the powers of the estimation error, the transmit distortion, the receive distortion,  and the interference  coming from other BSs across all the tiers.
\end{theorem}
\begin{proof}
See Appendix~\ref{CoverageProbabilityproof}.
\end{proof}

The calculation of the coverage probability is based on  the Laplace transforms obtained  in terms of  Lemma~\ref{LaplaceTransformGamma} and  Proposition~\ref{LaplaceTransform} as provided below. 
\begin{lemma}\label{LaplaceTransformGamma} 
 The Laplace transforms of the random variables given by~\eqref{error},~\eqref{transmit}, and~\eqref{receive}, expressing the estimation error,  $I_{ \etv_{w\mathrm{t}\mathrm{z}}}$, and $I_{ \etv_{w\mathrm{r}\mathrm{z}}}$, respectively, are given by
 \begin{align}
\mathcal{L}_{E_{w\mathrm{\mathrm{z}}}}\!\left(s_{w\mathrm{z}} \right)&=\frac{1}{\left( 1+\delta_{w}^{-2} \left(1+\kappa_{\mathrm{t}_{\mathrm{BS}w}} ^{2}\right)\sigma_{\hat{\bee}_{k}}^{2} \zeta_{w\mathrm{z}}s_{w\mathrm{z}} \right)^{K_{w}}}\\
\mathcal{L}_{I_{\etv_{t}w\mathrm{z}}}\!\left(s_{w\mathrm{z}} \right)&=\frac{1}{\left( 1+\kappa_{\mathrm{t}_{\mathrm{BS}w}}^{2}\zeta_{w\mathrm{z}} s_{w\mathrm{z}} \right)^{N_{w}}}, \\
\mathcal{L}_{I_{\etv_{r}w\mathrm{z}}}\!\left(s_{w\mathrm{z}} \right)&=\frac{1}{\left( 1+\kappa_{\mathrm{r}_{\mathrm{UE}w}}^{2} \zeta_{w\mathrm{z}}s_{w\mathrm{z}} \right)^{N_{w}}},
 \end{align}
 where $\zeta_{w\mathrm{z}}=\frac{\sigma_{\hat{\bh}_{w}}^{2}}{{T}_{w}\beta_{w\mathrm{z}}^{-1}}$.
\end{lemma}
\begin{proof}
In Appendix~\ref{SINRproof}, it is mentioned that $I_{ \etv_{w\mathrm{t}\mathrm{z}}}$ and $I_{ \etv_{w\mathrm{r}\mathrm{z}}}$ are  scaled Gamma distributions.
Hence, their Laplace transforms can be easily obtained. In a similar way, the Laplace transforms of the estimation error is also derived. 
\end{proof}
 \begin{proposition}\label{LaplaceTransform} 
Given that the typical user is associated with a LOS BS, the Laplace transform of the interference power in a  general realistic cellular network with randomly distributed multiple antenna BSs, operating at mmWave frequencies, having RATHIs and imperfect CSIT is given by
\begin{align}
&\mathcal{L}_{I_{\mathrm{L}}}\left(  s_{w } \right)=\prod_{j \in \mathcal{W}}\prod_{k=1}^{4}e^{\left( -{2\pi\lambda_{w}}b_{wk}\left( V_{jk}\left(x\right)+W_{jk}\left(x\right) \right)\right)},
\end{align}
where  $V_{jk}\left( x \right)=\int_{x}^{\infty} \left( 1-\frac{1}{\left( 1+\frac{\bar{f}_{wk}C_{\mathrm{L}_{j}}s_{w }}{K_{j}}\left( \frac{y}{t} \right)^{-\al_{\mathrm{L}}} \right)^{K_{j}}}  \right)p\left( t \right)\mathrm{d}t$,
$W_{jk}\left( x \right)=\int_{\psi_{j}(x_{j})}^{\infty}\! \left(\!\! 1\!-\!\frac{1}{\left( 1+\frac{\bar{f}_{wk}C_{\mathrm{N}_{j}}s_{w }}{K_{j}}\left( \frac{y}{t} \right)^{-\al_{\mathrm{N}}} \right)^{K_{j}}}  \!\!\right)\!\!\left( 1\!-\!p\left( t \right) \right)\mathrm{d}t$, $\bar{f}_{wk}=\frac{a_{wk}}{M_{rw}M_{tw}}$, $p\left( t \right)$ is the LOS PDF as well as $a_{wk}$ and $b_{wk}$ are constants defined in Table~I.
\end{proposition}
\begin{proof}
See Appendix~\ref{LaplaceTransformproof}.
\end{proof}

When a typical user is associated with an NLOS BS, the interference $\mathcal{L}_{I_{\mathrm{N}}}\left(  s_{w } \right)$ is obtained by means of a similar proposition.

These results are more general  than~\cite{PapazafComLetter2016,Papazafeiropoulos2017} for several reasons. For example, one reason,  regarding the former reference, is the realistic consideration of imperfect CSIT, while, with a comparison to the latter reference, we consider a mmWave transmission with its special characteristics.

\subsection{ASE }\label{AverageAchievableRate1}
Herein, we present the other main result of this paper,  which is unique in the research area of practical mmWave systems with hardware impairments, when the BSs are randomly positioned. Specifically, we refer to the ASE for a typical user served with mmWave transmission, while both the CSIT and the transceiver hardware are imperfect due to channel aging and RATHIs, respectively. The presentation is concise to avoid any repetition since the analysis and some definitions are similar to Section~\ref{coverage}.

We start with the definition of the ASE for a multi-tier setup. Specifically, we have~\cite{Dhillon2013}
\begin{align}
 \eta\left(   q_{w}, T_{w}  \right)=\underset{w \in \mathcal{W}}{\sum}K_{w}\lambda_{w}\log_{2}\left( 1+T_{w} \right)P_{\mathrm{c}}^{w}\left(   q_{w}, T_{w}  \right),\label{rate1}
\end{align}
where $P_{\mathrm{c}}^{w}$ is the  coverage probability conditional on the serving BS found in the $w$th tier given Theorem 1.
Similar to~\cite{Dhillon2013}, we do not derive the per tier coverage probability. However, we can assume that the ASE is  $P_{\mathrm{c}}^{w}=P_{\mathrm{c}}~\forall w$\footnote{The theorem holds for the cases of SISO transmission when the target SIRs are the same for all tiers and for SDMA transmission under the same assumption plus when the number of  BS antennas are the same for all tiers.}. 
\begin{theorem}\label{AverageAchievableRate}
 The downlink area spectral  efficiency of a realistic  MU-MIMO HetNet operating in the mmWave  frequencies in the presence of the inevitable RATHIs and channel aging, is
\begin{align}
 \eta\left(   q_{w}, T_{w}  \right)=P_{\mathrm{c}}\left(   q_{w}, T_{w}  \right)\underset{w \in \mathcal{W}}{\sum}K_{w}\lambda_{w}\log_{2}\left( 1+T_{w} \right),\label{rate1}
 \end{align}
 where $P_{\mathrm{c}}\left(   q_{w}, T_{w}  \right)$ is given by Theorem $1$. The various parameters, concerning this result, are also defined in Subsection~\ref{coverage} by means of Theorem~\ref{theoremCoverageProbability}.
\end{theorem}

 \section{Numerical Results}  \label{Numerical} 
This section presents illustrations of the analytical expressions, which are also verified by Monte Carlo simulations. We examine the impact of various design parameters such as the channel aging, RATHIs, numbers of users and BS antennas on the general theoretical expressions describing the coverage probability and the ASE provided by Theorems~\ref{theoremCoverageProbability} and~\ref{AverageAchievableRate}, respectively.
We choose a sufficiently large area of $5~\mathrm{km}\times 5$ $\mathrm{km}$, where the locations of different classes of BSs are simulated as realizations of different PPPs with given densities $\lambda_{\mathrm{B}_{w}}$. The users' PPP densities in all tiers are considered to be  $\lambda_{\mathrm{K}_{w}}=6\lambda_{\mathrm{B}_{w}}$. For the ease of exposition, we focus on a two-tier setup, where $\lambda_{\mathrm{B}_{2}}=0.5\lambda_{\mathrm{B}_{1}}$. Note that the simulation takes place over a finite window, while the analytical expressions rely on the assumption of an infinite plane, which results in border effects, being visible at high SNR. The setup under consideration includes BSs of $M=5$ number of antennas serving $K=2$ users.  The LOS and NLOS path-loss exponents are set to $\al_{\mathrm{L}_{w}}=3$ and $\al_{\mathrm{N}_{w}}=4$, respectively. The average downlink transmit power is $\rho_{w}^{\mathrm{BS}}=5~\mathrm{dBW}$. The network operates at $50$ GHz, while the bandwidth allocated for each user is $100$ MHz. The parameter $\beta$, included in the LOS probability function is given by $1/\beta=141.4$ meters. Moreover, the transmit antenna pattern is assumed to be $G_{20~\mathrm{dB},~0~dB,~30\degree}$~\cite{Bai2015}.

Remarkably, even in the special case of perfect CSIT, but still, time-varying, there is no known result in the literature studying RATHIs in the mmWave area. Moreover, if we assume a static channel, i.e., $\delta_{w}=1$, again, there is no known reference corresponding to this model since RATHIs and ATN are present. However, especially, when there is no channel aging,  all the impairments are set to zero and we consider only one tier, i.e., $w=1$, Theorem 1 coincides with \cite[Theorem~1]{Bai2015}, but there is no study regarding the ASE.


The figures depicting the proposed analytical expressions of the coverage probability $p_{c}\left(   q_{w}, T_{w}  \right)$ and the achievable user rate $\eta\left(   q_{w}, T_{w}  \right)$ are plotted  along with the corresponding simulated results. The ``solid''  lines with certain patterns illustrate the proposed analytical results with specific quality of imperfect CSIT, RATHIs, and user mobility, while the ``dot'' lines, in most of the figures, correspond to the ``ideal'' expressions with no transceiver impairments and no relative user movement. Similarly, the ``stars'' represent the simulation results. Apparently, the unavoidable effects under study, i.e., the RATHIs and the time variation of the channel degrade the system performance since the corresponding terms are met in the denominator of the SDINR. In fact, by increasing the severity of these effects, the performance worsens. Below,  for the sake of exposition and and without any impact on the following extracted conclusions, we assume that the transceivers of both tiers employ the same hardware. As a result, the parameters  defining the quality of the hardware are equal in both tiers. Also, for the sake of  simplification, we consider that, in general, all the other parameters in the two tiers are the same, e.g., $\delta_{1}=\delta_{2}=\delta$.
\subsection{Impact of Channel Aging}
In order to focus only on the impact of channel aging, we set all the additive impairments equal to zero and the ATN equal to $\sigma^{2}$ i.e., $\kappa_{\mathrm{t}_\mathrm{UE}w}=\kappa_{\mathrm{r}_\mathrm{UE}w}=0$ and $\xi=\sigma^{2}$. Hence, in Fig.~\ref{Fig1}, we depict the variation of the coverage probability for varying values of $\delta$.  Obviously, increasing $\delta$, being equivalent to lower mobility, the coverage probability increases. If the normalized variable  $f_{D}T_{s}$ becomes high enough, which means high mobility, then $\delta$ becomes very small, and the coverage probability is  inadequate to support any service. In other words, by focusing on the normalized variable $f_{D}T_{s}$, we observe that its increase degrades  $p_{c} $. Interestingly, when $f_{D}T_{s} \to 0.4$ and given this channel aging model, $f_{D}T_{s} \to 0$, i.e., there is no user mobility and the coverage probability decreases to minimum. Further increase of $f_{D}T_{s}$ results in improvement of the coverage probability. Indirectly, this behavior is explained by noticing that $\delta$ follows the variation of the Jakes autocorrelation function with $f_{D}T_{s}$ as an argument.
\begin{figure}[!h]
 \begin{center}
 \includegraphics[width=0.8\linewidth]{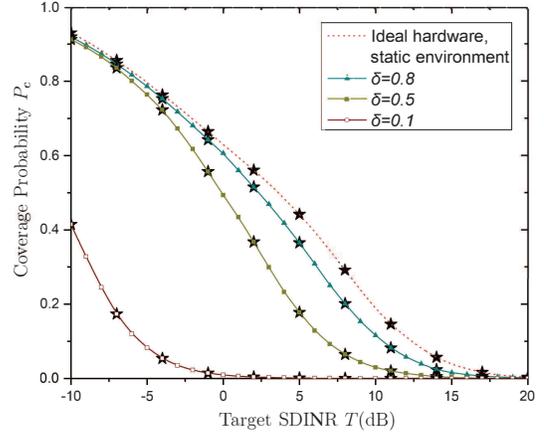}
 \caption{\footnotesize{Coverage probability of a MU- MIMO mmWave HetNet versus the target SDINR $T$ for varying  severity of channel aging, while the RATHIs and ATN are assumed to have no impact ($\kappa_{\mathrm{t}_{\mathrm{BS}}}=\kappa_{\mathrm{r}_{\mathrm{UE}}}=0$,  $\xi=\sigma^{2}$).}}
 \label{Fig1}
 \end{center}
 \end{figure}
  \begin{figure}[!h]
 \begin{center}
 \includegraphics[width=0.8\linewidth]{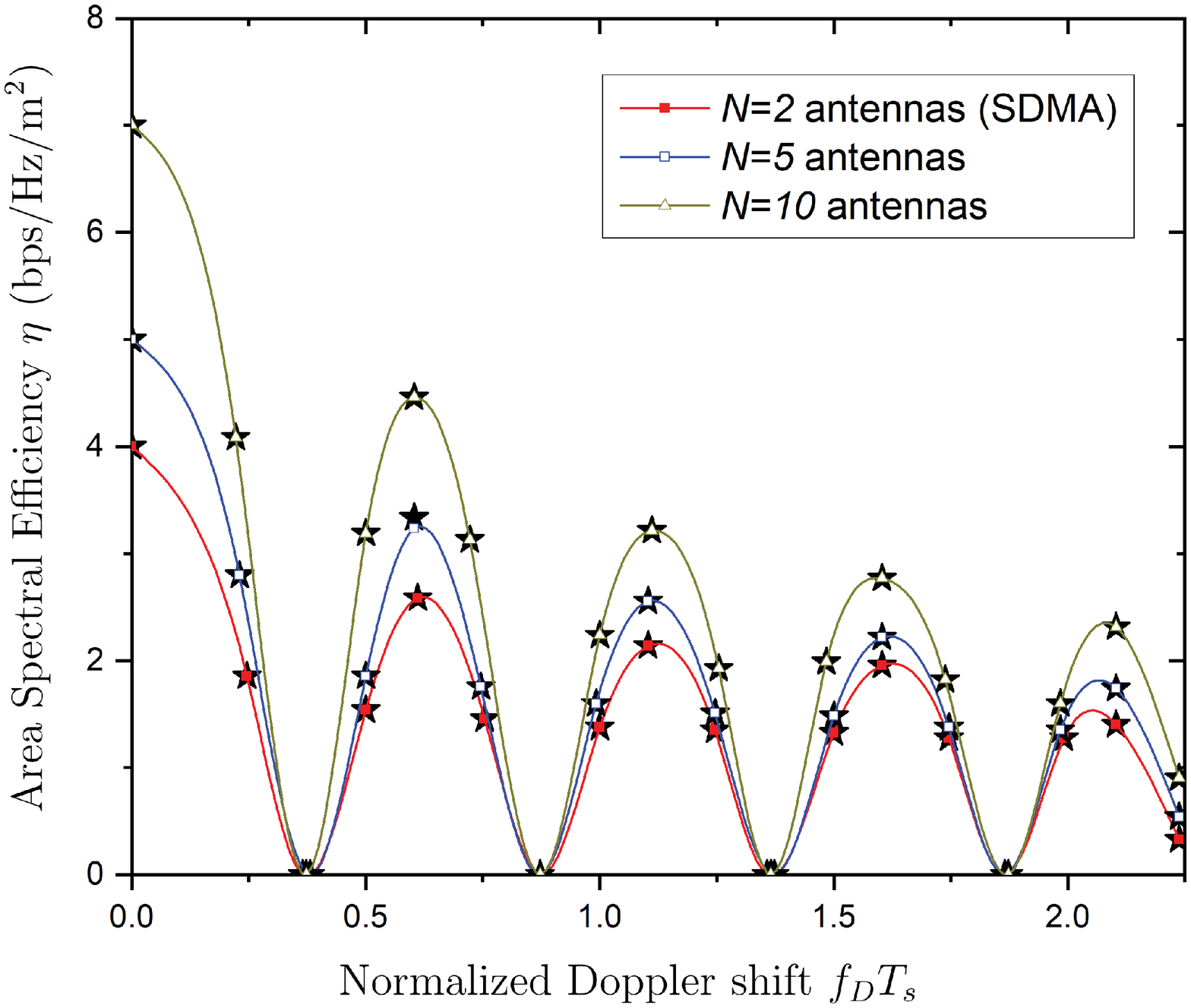}
 \caption{\footnotesize{Area spectral efficiency of a MU-MIMO mmWave HetNet versus the target normalized Doppler shift $f_{D}T_{s}$ for varying number of BS antennas, while the RATHIs and ATN are assumed to have no impact ($\kappa_{\mathrm{t}_{\mathrm{BS}}}=\kappa_{\mathrm{r}_{\mathrm{UE}}}=0$,  $\xi=\sigma^{2}$).}}
 \label{Fig2}
 \end{center}
 \end{figure}
 
Regarding Fig.~\ref{Fig2}, it depicts the ASE versus the normalized Doppler shift $f_{\mathrm{D}}T_{\mathrm{s}}$ for varying number of BS antennas. Especially, in the case $N=2$, then we have space division multiple access (SDMA) since we have considered $K=2$ number of users. It is obvious that the channel aging (lower quality of CSIT) decreases the downlink ASE to zero several times. In fact, the curves have some ripples with their behavior following the shape of the $\mathrm{J}_{0}(\cdot)$, i.e., the zeroth-order Bessel function of the first kind. Hence, we note that the ASE increases up to a point and then decreases to zero, again and again, tending finally to zero. The basic shape of the plots does not change with the number of BS antennas. Hence, the zero points appear at specific values of the normalized Doppler shift. This is justified because $\delta=r[1]$, depending only on $f_{\mathrm{D}}T_{\mathrm{s}}$. Only the magnitude increases with increasing $N$ since more degrees of freedom are provided.
 
\subsection{Impact of RATHIs}
In Fig.~\ref{Fig3}, we plot the coverage probability as a function of the target SDINR for varying values of the RATHIs, but with no channel aging and no ATN. These nominal values of RATHIs are quite reasonable since these values can be found in practical systems. Specifically, according to~\cite{Bjornson2015,Papazafeiropoulos2017}, if we assume that we have an Analog-to-Digital Converter (ADC) quantizing the received signal to a $b$ bit resolution, then $\kappa_{\mathrm{t}_{\mathrm{BS}}}=\kappa_{\mathrm{r}_{\mathrm{UE}}}=2^{-b}/\sqrt{1-2^{-2b}}$.  This expression for 2, 3, and 4 bits gives $\kappa_{\mathrm{t}_{\mathrm{BS}}}=0.258,~0.126,$ and $0.062$, respectively. The chosen number of bits concerns the trend to employ low-precision ADCs in 5G networks~\cite{Mo2017,Roth2017}. Clearly, the number of bits defines the impact on the coverage probability. In fact, the smaller the resolution, the higher the degradation. Moreover, as the transmit power increases, the degradation from the RATHIs becomes more severe because these are power-dependent. Actually, as we characterize conventional systems with interference as interference-limited at high SNR, in practice, they are also RATHIs-limited. In other words, the impact of RATHIs becomes more severe at high SNR. Fortunately, next-generation systems such as massive MIMO are supposed to work with very low transmit power. However, during the system design, special attention has to be taken, in order not to exceed the specifications of the system.

Fig.~\ref{Fig4} illustrates the variation of the ASE with the SDINR for varying RATHIs. Clearly, an increase of the additive hardware impairments brings a decrease to the ASE as expected. Moreover, at the same figure, we show the ASE for ideal hardware and no channel aging. The great gap between the curves appears at high SDINR since at this regime, the RATHIs start having an impact (RATHIs are power dependent). For the same reason, the variation of RATHIs becomes more distinguishable. Furthermore, we demonstrate a comparison between the impacts transmit and receive additive impairments. Notably, the transmit impairments result in a higher loss. This is reasonable since the additive impairments affect also the term including the estimation error as can be seen by~\eqref{error}.

\begin{figure}[!h]
 \begin{center}
 \includegraphics[width=0.8\linewidth]{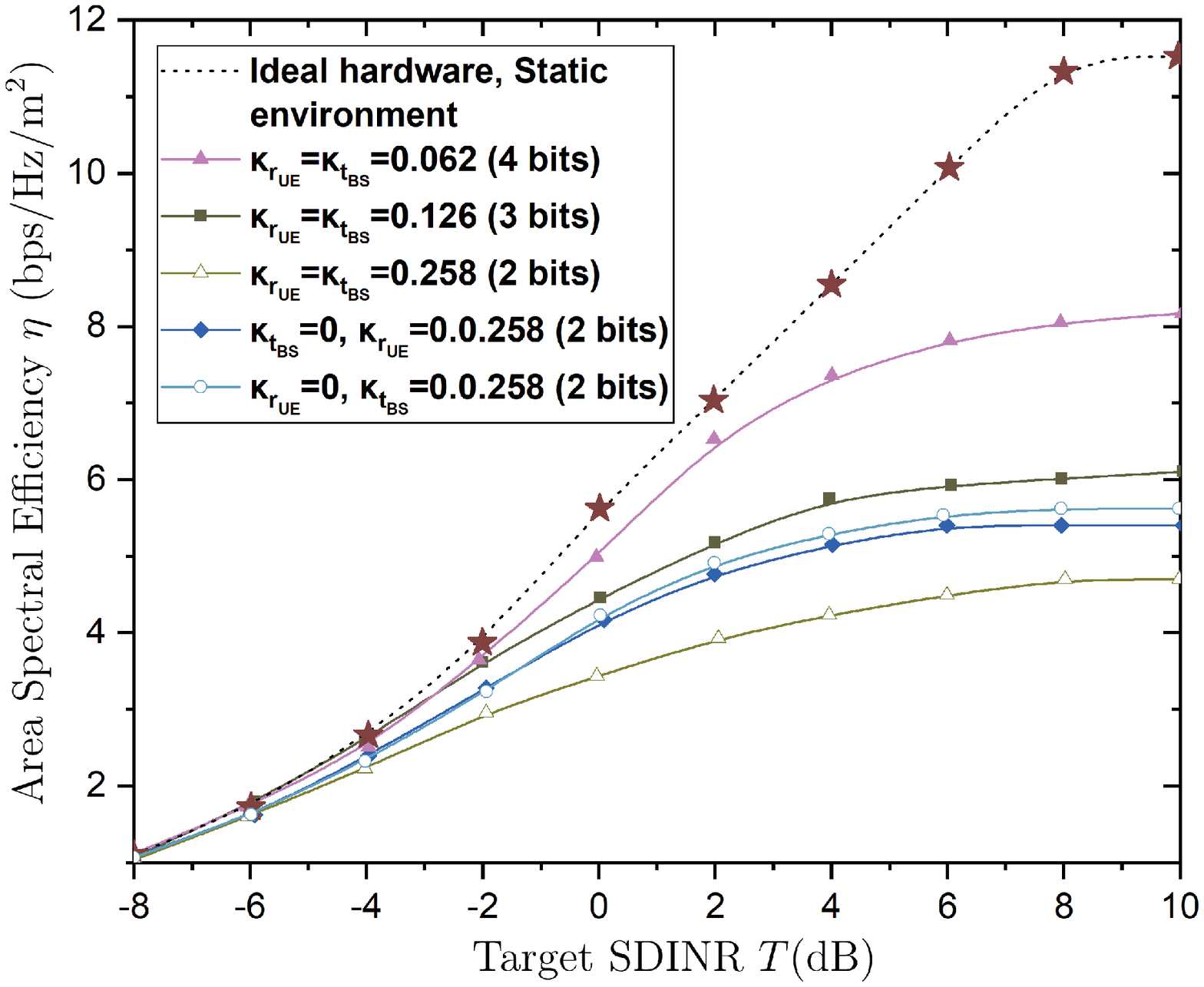}
 \caption{\footnotesize{Coverage probability of a MU-MIMO mmWave HetNet versus the target SDINR $T$ for varying  severity of the RATHIs, while the channel aging and ATN are assumed to have no impact ($\delta=1$,  $\xi=\sigma^{2}$).}}
 \label{Fig3}
 \end{center}
 \end{figure}
 \begin{figure}[!h]
 \begin{center}
 \includegraphics[width=0.8\linewidth]{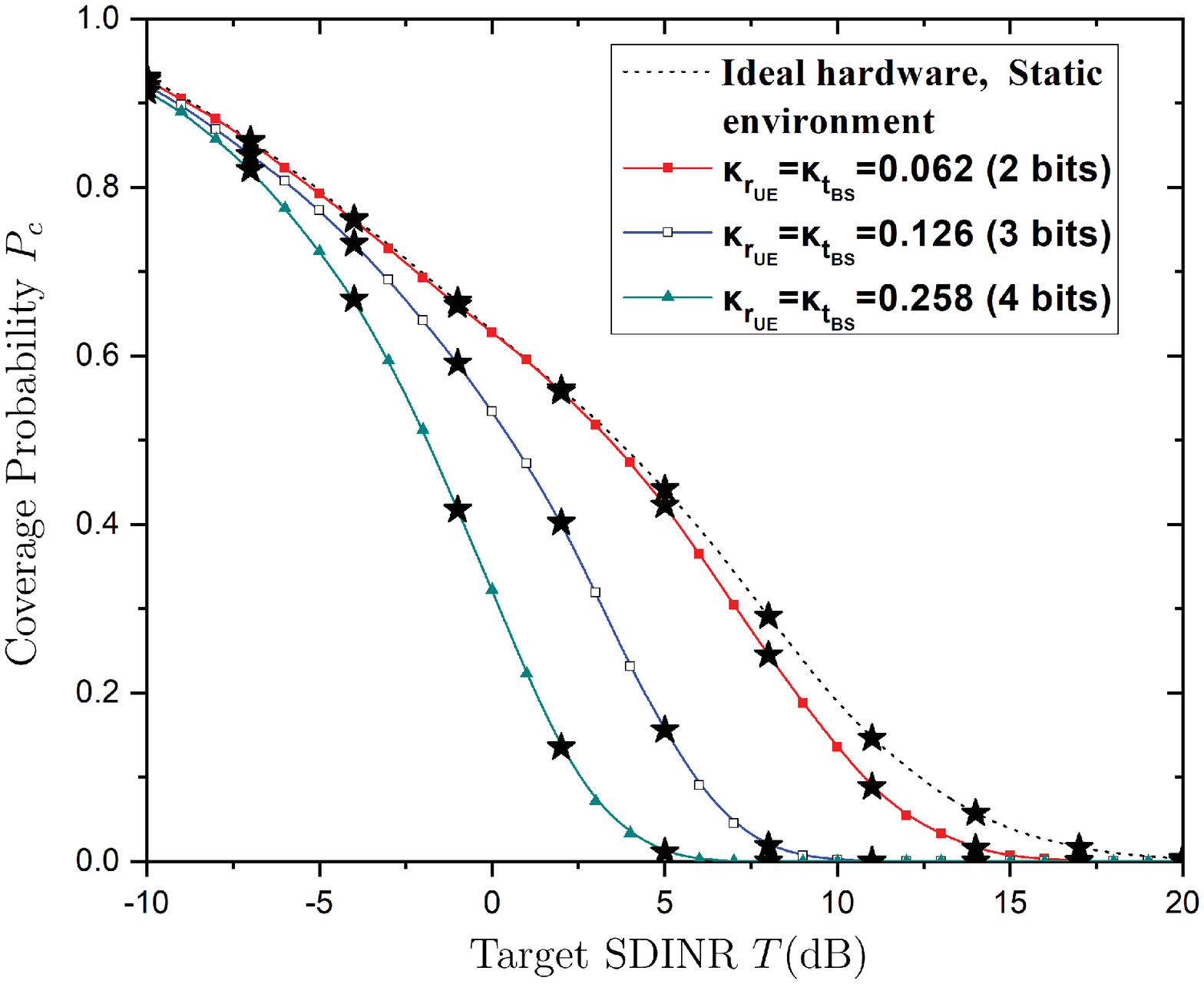}
 \caption{\footnotesize{Area spectral efficiency of a MU-MIMO mmWave HetNet versus the target SDINR $T$ for varying  severity of the RATHIs, while the channel aging and ATN are assumed to have no impact ($\delta=1$,  $\xi=\sigma^{2}$).}}
 \label{Fig4}
 \end{center}
 \end{figure}
\subsection{Impact of ATN}
Fig.~\ref{Fig5} illustrates the impact of ATN on the coverage probability. The selected values are multiples of $\xi=1.6\sigma^{2}$, which has been borrowed from \cite{Bjornson2015}. Specifically, in \cite{Bjornson2015}, the authors assumed a low  noise amplifier with $F$ being the noise amplification factor.   Hence, if we assume that $F=2$ dB and $b=3$ bits, then $\xi=\frac{F \sigma^{2}}{1-2^{-2b}}=1.6\sigma^{2}$. As can be seen, ATN affects the coverage probability, but less than the other factors under consideration in this work. Notably, at high SNR, the system becomes power-limited, and thus, the ATN does not make the coverage probability change. As a result, 
the lines coincide at high SNR.
\begin{figure}[!h]
 \begin{center}
 \includegraphics[width=0.8\linewidth]{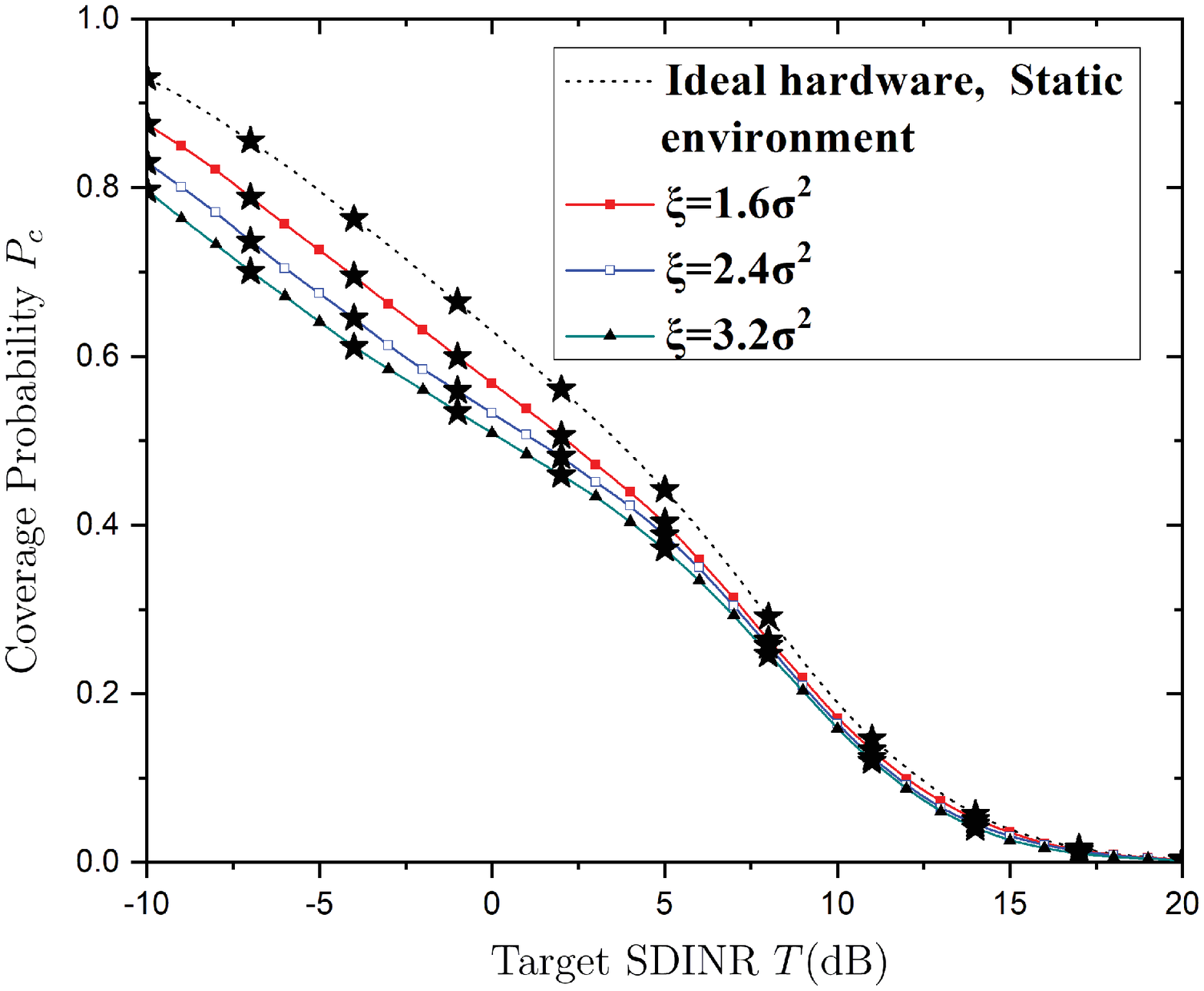}
 \caption{\footnotesize{Coverage probability of a MU-MIMO mmWave HetNet versus the target SDINR $T$ for varying quality of the ATN, while the RATHIs and channel aging are assumed to have no impact ($\kappa_{\mathrm{t}_{\mathrm{BS}}}=\kappa_{\mathrm{r}_{\mathrm{UE}}}=0$,  $\delta=1$).}}
 \label{Fig5}
 \end{center}
 \end{figure}
  \begin{figure}[!h]
 \begin{center}
 \includegraphics[width=0.8\linewidth]{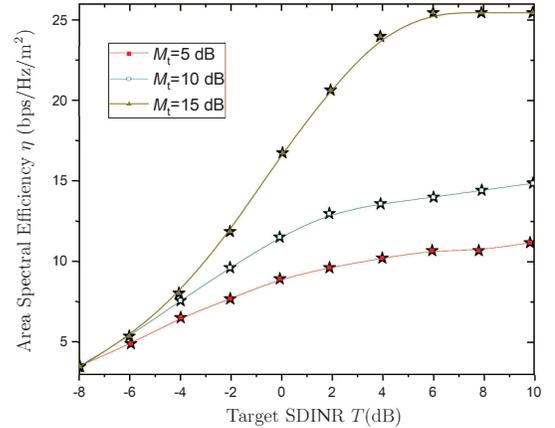}
 \caption{\footnotesize{  Area spectral efficiency of a MU-MIMO mmWave HetNet versus the target SDINR $T$ for varying the main lobe directivity gain $M_{\mathrm{t}}$ ($\kappa_{\mathrm{t}_{\mathrm{BS}}}=\kappa_{\mathrm{r}_{\mathrm{UE}}}=0.126$, $\xi=1.6\sigma^{2}$, and  $\delta=0.7$).}}
 \label{Fig6}
 \end{center}
 \end{figure}
 \subsection{Impact of Transmit Diversity Gain}
In Fig.~\ref{Fig6}, we have plot the ASE for  different values of the main lobe transmit diversity gain $M_{t}$. As $M_{t}$ increases, the ASE increases since the beams of mmWaves appear higher directivity. In fact, by doubling $M_{t}$,  the ASE increases by $2.6~\mathrm{bps/Hz/m^{2}}$, while by tripling $M_{t}$, the ASE increases almost 14 units. Hence, the increase of the directivity gain of the main lobe offers high positive contribution to the ASE. The study of the other parameters concerning the antennas arrays is left for future work. 
\subsection{Impact of Numbers of Users and Antennas}
Figs.~\ref{Fig7} and~\ref{Fig8}, assuming certain values for  the hardware impairments and channel aging, shed light into the impact of users and BS antennas, respectively. It is important to mention that $K$ and $M$ affect the severity of the hardware impairments as can be seen from Proposition 1 that relates the various terms with $K$ and $M$. In particular, in Fig. 7, having set $K=2$ and increasing the number of BS antennas, we observe an improvement in terms of the  coverage probability as expected. On the other 
hand, Fig. 8 shows the impact of increasing $K$ on the coverage probability, i.e., $p_{c} $ decreases and gets the lowest value in the case of space division multiple access (SDMA).  These results have been already reported from other similar works~\cite{Dhillon2013,PapazafComLetter2016,Papazafeiropoulos2017}.  In fact, we confirm that serving less users is preferable.
\begin{figure}[!h]
 \begin{center}
 \includegraphics[width=0.8\linewidth]{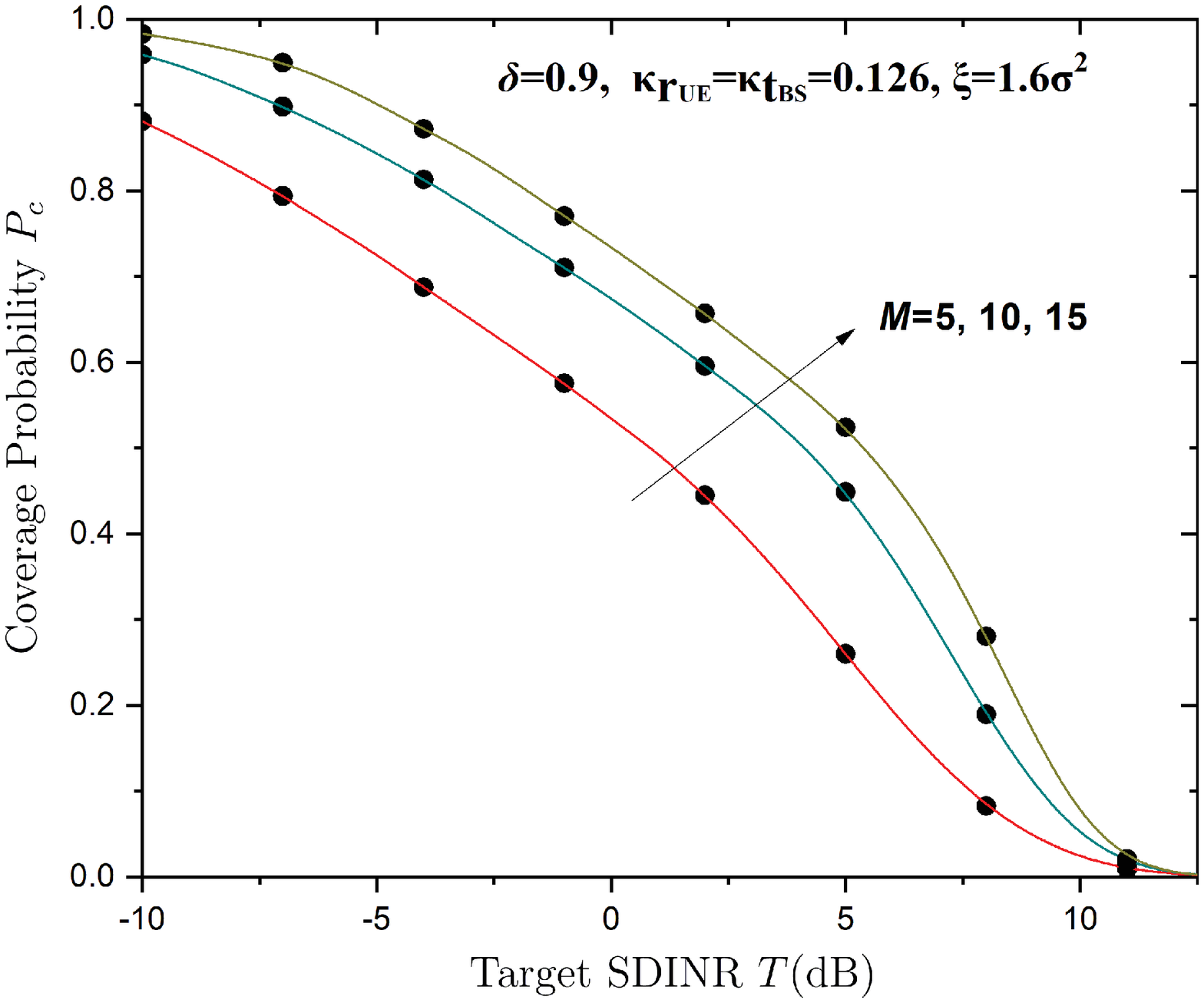}
 \caption{\footnotesize{  Coverage probability of a MU-MIMO mmWave HetNet versus the target SDINR $T$ for varying  number of BS antennas $M$ ($\kappa_{\mathrm{t}_{\mathrm{BS}}}=\kappa_{\mathrm{r}_{\mathrm{UE}}}=0.126$, $\xi=1.6\sigma^{2}$, and  $\delta=0.9$).}}
 \label{Fig7}
 \end{center}
 \end{figure}
 \begin{figure}[!h]
 \begin{center}
 \includegraphics[width=0.8\linewidth]{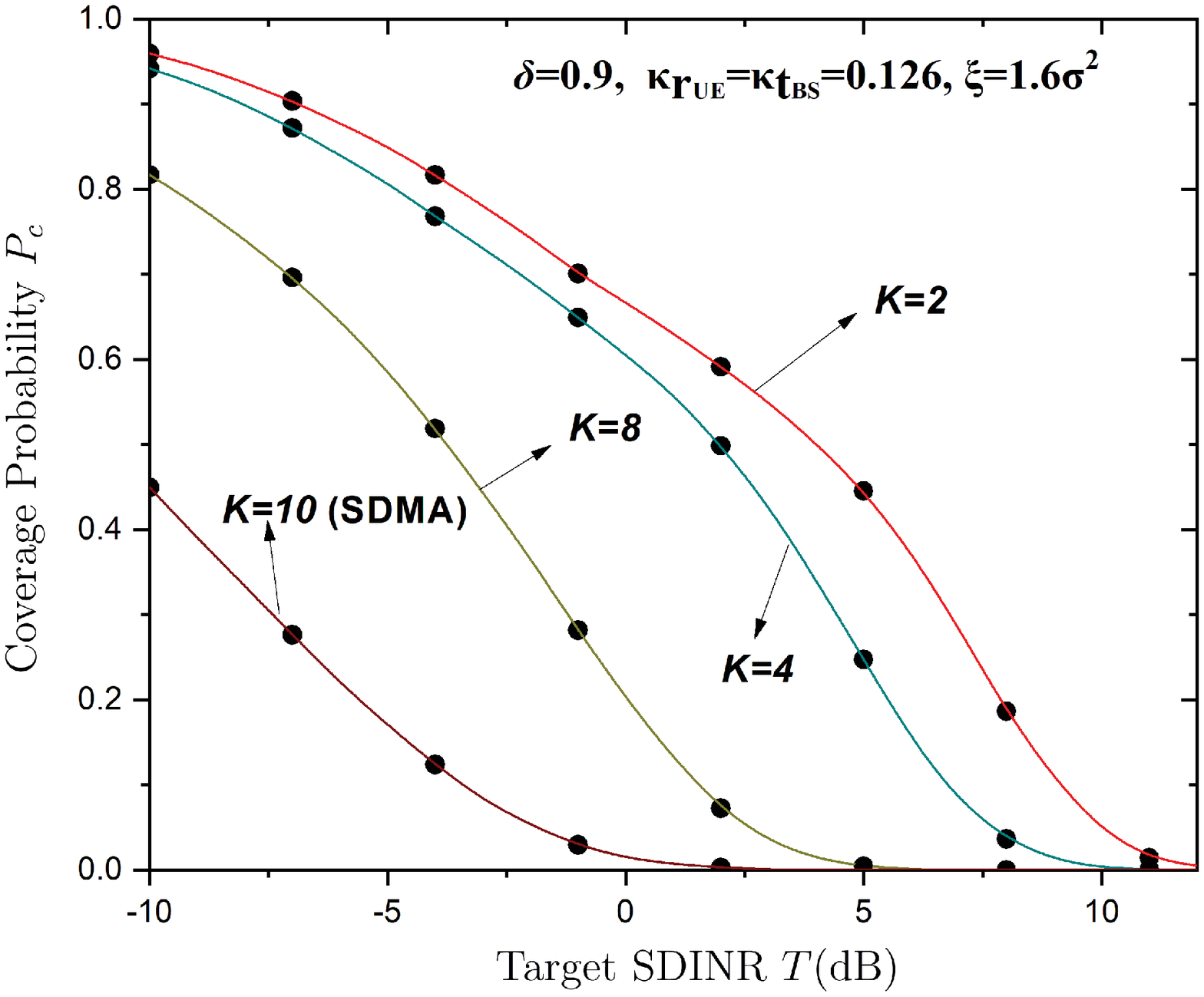}
 \caption{\footnotesize{  Coverage probability of a MU-MIMO mmWave HetNet versus the target SDINR $T$ for varying  number of users  $K$ ($\kappa_{\mathrm{t}_{\mathrm{BS}}}=\kappa_{\mathrm{r}_{\mathrm{UE}}}=0.126$, $\xi=1.6\sigma^{2}$, and  $\delta=0.9$).}}
 \label{Fig8}
 \end{center}
 \end{figure}

\section{Conclusion} \label{Conclusion} 
In this paper, we proposed a novel general framework to model realistic multi-tier MU-MIMO downlink  mmWave networks with PPP distributed BSs. First, we obtained the SDINR, and then, we derived the coverage probability as well as the ASE. As far as the authors are aware, this is the first work introducing the mmWave transmission in a realistic MIMO HetNet, where the RTHIs and channel aging are taken into account. The assessment of the practical potential of such system resulted in a comprehensive understanding.  The consideration of aggregation or repulsion e.g., in terms of Cox or Gibbs processes, respectively, is an interesting topic for future research in the areas of both mmWave transmission and impairments. Specifically,  in this work, we showed in terms of numerical results that  directional beamforming with sectored antenna  is preferable, while when the directivity collapses, the system does not behave efficiently. Moreover, we examined the impact of RATHIs and ATN, and we depicted that the former and the latter impairments become important at high SNR and low SNR, respectively. Another interesting observation is the demonstration of the degradation of the system by increasing user mobility.

\begin{appendices}
\section{Useful Lemmas~\cite{Bai2015}}\label{UsefulLemmas}
\begin{lemma}
Assuming that the typical user has at least one LOS BS, the conditional PDF of its distance to the nearest LOS BS is
\begin{align}
f_{\mathrm{L}}\left( x_{w} \right)=2\pi \lambda_{w} x_{w} p\left( x_{w} \right)\mathrm{e}^{-2\pi \lambda_{w}\int_{0}^{x_{w}}rp\left( r \right)\mathrm{d}r}/B_{\mathrm{L}_{w}},
\end{align}
where $B_{\mathrm{L}_{w}}=1-\mathrm{e}^{-2\pi \lambda_{w} \int_{0}^{\infty}rp\left( r \right)\mathrm{d}r}$ is the probability that a user observes at least one LOS BS, and $p\left( r \right)$ is the LOS PDF. In a similar case, given the
user has at least one NLOS BS, the conditional
PDF of the distance to the nearest NLOS BS is
\begin{align}
\!\!\!\!\!\!f_{\mathrm{N}}\!\left( x_{w} \right)\!=\!2\pi \lambda_{w} x_{w} \!\left( 1\!-\!p\left( x_{w} \right)\! \right)\!\mathrm{e}^{-2\pi \lambda_{w}\!\int_{0}^{x_{\!w}}\!\!r\left( 1-p\left( r \right)\! \right)\mathrm{d}r}\!\!/\!B_{\mathrm{N}_{w}},\!\!\!\!
\end{align}
where $B_{\mathrm{N}_{w}}=1-\mathrm{e}^{-2\pi \lambda_{w} \int_{0}^{\infty}r \left( 1-p\left( r \right) \right)\mathrm{d}r}$ is the probability that a user observes at least one NLOS BS.
\end{lemma}

\begin{lemma}\label{lemma3}
The probability that the user is associated with a
LOS BS is
\begin{align}
 A_{\mathrm{L}_{w}}=B_{\mathrm{L}_{w}}\int_{0}^{\infty}\mathrm{e}^{-2\pi\lambda_{w}\int_{0}^{\Psi_{\mathrm{L}}\left( x \right)}\left( 1-p\left( t \right) \right)t\mathrm{d}t}f_{\mathrm{L}}\left( x \right)\mathrm{d}x,
\end{align}
where $\Psi_{\mathrm{L}}\left( x_{w} \right)=\left( C_{\mathrm{N}}/C_{\mathrm{L}} \right)^{1/\al_{\mathrm{N}_{w}}}x_{w}^{\al_{\mathrm{L}_{w}}/{\al_{\mathrm{N}_{w}}}}$, while the probability that the user is associated with a NLOS BS is $A_{\mathrm{N}_{w}}=1-A_{\mathrm{L}_{w}}$.
\end{lemma}
Assuming that a user is associated with a LOS BS, the PDF of the distance to its serving BS is
\begin{align}
\!\!\!\!\!\hat{f}_{\mathrm{L}}\left( x_{w} \right)=\frac{B_{\mathrm{L}_{w}}  f_{\mathrm{L}}\left( x_{w} \right)}{A_{\mathrm{L}_{w}}}\mathrm{e}^{-2\pi \lambda_{w}\int_{0}^{\Psi_{\mathrm{L}}\left( x_{w} \right)}\left( 1-p\left( t \right) \right)t\mathrm{d}t}/B_{\mathrm{L}_{w}}.
\end{align}
Taking for granted that the user can be served by a NLOS BS, the PDF of the distance to its serving is
\begin{align}
\hat{f}_{\mathrm{N}}\left( x_{w} \right)=\frac{B_{\mathrm{N}_{w}}  f_{\mathrm{N}}\left( x_{w} \right)}{A_{\mathrm{L}_{w}}}\mathrm{e}^{-2\pi \lambda_{w}\int_{0}^{\Psi_{\mathrm{N}}\left( x_{w} \right)}p\left( t \right) t\mathrm{d}t}/B_{\mathrm{N}_{w}},
\end{align}
where $\Psi_{\mathrm{N}}\left( x_{w} \right)=\left( C_{\mathrm{L}}/C_{\mathrm{N}} \right)^{1/\al_{\mathrm{L}_{w}}}x_{w}^{\al_{\mathrm{N}_{w}}/{\al_{\mathrm{L}_{w}}}}$.
\section{Proof of
Proposition~\ref{SINR}}\label{SINRproof}
Let us first define $\hat{\bV}_{w,n}=\bar{\bH}^{\H}_{w,n}\left( \bar{\bH}_{w,n} \bar{\bH}^{\H}_{w,n} \right)^{-1}$.
The PDF of the desired signal power, being in the numerator of the SDINR in~\eqref{SDIR1},  is $\Gamma\left( \Delta_{w},\sigma_{\hat{\bh}_{w}}^{2} \right)$ distributed with $\Delta_{w}=N_{w}-K_{w}+1$ because it can be written as
\begin{align}
Z_{w,n}=|\bar{\bh}_{w,n-1}^{\H}\bv_{w,n-1}|^{2}\cdot \|\bh_{w,n-1}\|^{2}. 
\end{align}
As can be seen, $Z_{w,n}$ is written as the product of two independent random variables distributed as $B\left( N_{w}-K_{w}+1,K_{w}-1 \right)$ and $\Gamma\left( N_{w},\sigma_{\hat{\bh}_{w}}^{2} \right)$,  respectively~\cite{Jindal2006}\footnote{ An equivalent description can be made by means of  the Erlang distribution with  shape and  scale parameters $\Delta_{w}$ 
and $\sigma_{\hat{\bh}_{w}}^{2}$, respectively}. Note that   $\|\hat{\bh}_{w,n-1}^{\H}\|^{2}\mathop \sim \limits^{\tt d}\Gamma[N_{w},\sigma_{\hat{\bh}_{w} }^{2}]$ since the random variable $\|\hat{\bh}_{w,n-1}^{\H}\|^{2}$  is the linear combination of $N_{w}$ i.i.d. exponential random variables each with variance  $\sigma_{\hat{\bh} }^{2}$. In a similar manner, the term including the error in the denominator can be in written as a sum of $K_{w}$ independent random variables   since 
\begin{align}
 \!\!\!\!\! \!E_{w,n}&=\beta_{w{\mathrm{z}}}\delta_{w}^{-2} \left(1+\kappa_{\mathrm{t}_{\mathrm{BS}}w} ^{2}\right)\big\|
                    \tilde{\bee}_{w,n}^{\H}\hat{\bV}_{w,n-1}
              \big\|^2\!\!\nn\\
              &=\!\! =\delta_{w}^{-2} \left(1+\kappa_{\mathrm{t}_{\mathrm{BS}}w} ^{2}\right)\sum_{i=1}^{K_{w}} \!\big|
                    \tilde{\bee}_{w,n}^{\H}\hat{\bv}_{wi,n-1}
              \big|^2\nn.
 \end{align}
Given that $\hat{\bv}_{wi,n-1}$ has unit norm  and  is independent of $ \tilde{\bee}_{w,n}$, $\big|  \tilde{\bee}_{w,n}^{\H}\hat{\bv}_{wi,n-1} \big|^2$ is the squared modulus of a linear combination of $K_{w}$ complex random variables   distributed as $\Gamma\left( 1,\sigma_{\hat{\bee}_{w} }^{2} \right)$. As a result, $\big\|\tilde{\bee}_{w,n}^{\H}\hat{\bV}_{w,n-1}\big\|^2$ is  $\Gamma\left( K_{w},\sigma_{\hat{\bee}_{w} }^{2} \right)$ distributed. Taking the expectation over the transmit and receive distortion noises, we have  $I_{ \etv_{w\mathrm{t}\mathrm{z},n}}=\beta_{w{\mathrm{z}}}\kappa_{\mathrm{t}_\mathrm{BS}w,n}^{2}\|{\bh}_{w,n}\|^{2}$ and $ I_{ \etv_{w\mathrm{r}\mathrm{z},n}}=\beta_{w{\mathrm{z}}}{\kappa_{\mathrm{r}_\mathrm{UE}w,n}^{2}}\|{\bh}_{w,n}\|^{2}$, which both follow a scaled  $\Gamma (N_{w},1)$ distribution. Moreover, the last term in the denominator, representing  the interference from other BSs $I_{\mathrm{z},n}$, can be written as $g_{jl,n}=|\bg_{jl,n}^{\H}\bV_{jl,n}|^2\sim \Gamma (K_{j},1)$ because $\bV_{jl,n}$ expresses  the precoding matrices of other BSs, having  unit-norm and being independent from  ${\bg}_{jl,n}$. Consequently, since $g _{jl,n} $  can be written as a linear combination of $K_{j}$ independent complex normal random variables with unit variance, we obtain that $g_{jl,n}\sim \Gamma (K_{j},1)$.
\section{Proof of Theorem~\ref{theoremCoverageProbability}}\label{CoverageProbabilityproof}
The proof starts with the description of $ p_{c}$ by the law of total probability. Specifically, taking into consideration that a user can be associated with a LOS or a NLOS BS, we have to provide the corresponding conditional coverage probabilities $ p_{c,\mathrm{L}}$ and $ p_{c,\mathrm{N}} $, respectively.  Hence, we have 
\begin{align}
p_{c}  =  p_{c,\mathrm{L}} +p_{c,\mathrm{N}}. 
\end{align}
The derivation of the conditional probabilities follows.
We start with the definition~\ref{def1} by focusing on the LOS BS case. We have
\begin{align}
 \!\!&p_{c,\mathrm{L}}\!\!\!=\!\EE\!\!\left[\!\mathds{1}\!\!\left( \underset{w \in \mathcal{W}}{\bigcup}	A_{\mathrm{L}_{w}}\underset{x_{w} \in \Phi_{L_{w}}}{\bigcup}\!\!\!\! \mathrm{SDINR}_{\mathrm{L}}\!\left(q_{w} , x_{w} \right)\!>\!\!\!T_{w} \right)\!\!  \right]\\
 &\le \EE\!\left[\! \underset{w \in \mathcal{W}}{\sum}	A_{\mathrm{L}_{w}}\underset{x_{w} \in \Phi_{L_{w}}}{\sum} \mathds{1}\left(\mathrm{SDINR}_{\mathrm{L}}\left( q_{w} ,x_{w} \right)>T_{w} \right)  \right]
 \label{coverage_definition0}\\
\! \!&=\!\!\! \underset{w \in \mathcal{W}}{\sum}A_{\mathrm{L}_{w}}\EE\!\left[	\underset{x_{w} \in \Phi_{L_{w}}}{\sum} \mathds{1}\left(\mathrm{SDINR}_{\mathrm{L}}\left( q_{w} ,x_{w} \right)>T_{w} \right)  \right]\label{coverage_definition}\\
\! \!&=\!\!\! \underset{w \in \mathcal{W}}{\sum}A_{\mathrm{L}_{w}}\EE\!\!\left[\!	\underset{x_{w} \in \Phi_{L_{w}}}{\!\!\sum}\!\!\!\! \mathbb{P}\left(
             Z_{w}>T_{w}\beta_{w\mathrm{L}}^{-1}\left( E_{w}+ 
           D_{w} +     I_{\mathrm{L}}  +\xi_{w,n}^{2} \right) \right) \! \right]\!
 \nn\\
\! \!&=\!\!\!\! \underset{w \in \mathcal{W}}{\sum}A_{\mathrm{L}_{w}}\EE\!\left[\int_{0}^{\infty}\! \mathbb{P}\left(
             Z_{w}\!>\!T_{w}\beta_{w\mathrm{L}}^{-1}\left( E_{w\mathrm{L}}\!+\! 
           D_{w\mathrm{L}}\! +\!
     I_{\mathrm{L}}+\xi_{w,n}^{2}   \right)\! \right) \!\! \right]\!\nn\\
     &~~~~~~~~\times\hat{f}_{\mathrm{L}}\!\left( x_{w} \right)\!\mathrm{d}x_{w}\label{coverage_definition100}
\end{align}
where $D_{w\mathrm{L}}=   I_{ \etv_{w\mathrm{t}\mathrm{L}}}
+I_{ \etv_{w\mathrm{r}\mathrm{L}}}$ expresses the total additive  distortion from both the transmitter and the receiver. Note that $E_{w}$ and $D_{w}$ do not depend on the BSs located at $R_{jl}$ far from the  typical user. In~\eqref{coverage_definition0}, we have applied  the Boole's inequality (union bound). Then, we substitute the SDINR in~\eqref{coverage_definition}. Next, we  employ the Campbell-Mecke Theorem~\cite{Chiu2013a}, and use the fact that
$Z_{w}$ is Gamma distributed with PDF given by $\mathbb{P}\left(Z_{w}>z  \right)=e^{-z}\displaystyle\sum_{i=0}^{\Delta_{w}-1}\frac{z^{i}}{i!}$. Hence, its PDF is provided by~\eqref{eq PDF1 1}. Now, 
\eqref{coverage_definition100} becomes
\begin{align}
 &\!\!\!p_{c,\mathrm{L}} \le \!\!\! \underset{w \in \mathcal{W}}{\sum}\!A_{\mathrm{L}_{w}}\int_{0}^{\infty} \!\!\!\EE\bigg[ e^{- \frac{T_{w} \beta_{w\mathrm{L}}^{-1}\left( E_w+D_{w}+ \xi^{2}\right)}{\sigma_{\hat{\bh}_{w}}^{2} } }e^{-\frac{{T}_{w}\beta_{w\mathrm{L}}^{-1} I_{\mathrm{L} }}{{\sigma_{\hat{\bh}_{w}}^{2} }}}\!\!\!\nn\\
&\times \!
\sum_{i=0}^{\Delta-1}\sum_{u=0}^{i}\!\binom{i}{u}\frac{\left(  {T}_{w}  \beta_{w\mathrm{L}}^{-1}\left( E_w\!+\!D_{w}\!+\!\xi^{2} \right)  \right)^{i-u}\!\left(     {T}_{w} \beta^{-1}_{w\mathrm{L}} I_{\mathrm{L}}  \right)^{u}}{i!\left( \sigma_{\hat{\bh}_{w}}^{2}  \right)^{i}}\bigg]\nn\\
&\times\hat{f}_{\mathrm{L}}\!\left( x_{w} \right)\mathrm{d}x_{w} \label{coverage1}\\
&
\!\!\!=\!\underset{w \in \mathcal{W}}{\sum}A_{\mathrm{L}_{w}}\sum_{i=0}^{\Delta_{w}-1}\!\sum_{u=0}^{i}\!\sum_{u_1+u_2+u_3+u_4=i-u}e^{-s_{w\mathrm{L}}\xi^{2}}\nn\\
&\times\int_{0}^{\infty}\!\!\!\binom{i}{u}\!\!\binom{i\!-\!u}{u_1+u_2+u_3+u_4}\frac{\left( -1 \right)^{i}\!s_{w\mathrm{L}}^{u}\left(  s_{w\mathrm{L}}\xi^{2} \right)^{u_{4}} }{i!}\!\nn\\
&\times\frac{\mathrm{d}^{u_1}}{\mathrm{d}s_{w\mathrm{L}}^{u_1}}\mathcal{L}_{E_{w\mathrm{z}} }\!\left( s_{w\mathrm{L}} \right)\!\!\! \frac{\mathrm{d}^{u_2}}{\mathrm{d}s_{w\mathrm{L}}^{u_2}}\mathcal{L}_{I_{\etv_{t}w\mathrm{L}}}\!\!\left(s_{w\mathrm{L}} \right)\frac{\mathrm{d}^{u_3}}{\mathrm{d}s_{w\mathrm{L}}^{u_3}}\mathcal{L}_{I_{\etv_{r}w\mathrm{L}}}\!\!\left(s_{w\mathrm{L}} \right)\! \nn\\
&\times\frac{\mathrm{d}^{u}}{\mathrm{d}s_{w\mathrm{L}}^{u}} \mathcal{L}_{I_{\mathrm{ L}}}\!\!\left(s_{w\mathrm{L}} \right)\hat{f}_{\mathrm{L}}\!\left( x_{w} \right)\mathrm{d}x_{w},\!\label{coverage3}
\end{align}
where  $A_{\mathrm{L}_{w}}$  in~\eqref{coverage1} is defined in Lemma~\ref{lemma3}, and we have applied the Binomial theorem. Setting  $ s_{w\mathrm{L}}=\frac{{T}_{w}\beta_{w\mathrm{L}}^{-1}}{\sigma_{\hat{\bh}_{w}}^{2}} $, and using the Multinomial theorem, we result in~\eqref{coverage3}. Actually, we have  taken  the inner sum  over all combinations of nonnegative integer indices $u_1$ to $u_3$ constraining the sum  $u_1+u_2+u_3$ to $i-u$. In addition, we have applied the definition of the Laplace transform $\mathbb{E}_{I }\left[  e^{-s I }\left( s I  \right)^{i}\right]=s^{i}\mathcal{L}\{t^{i}g_{I }\left( t \right)\}\left( s \right)$ and the Laplace identity $t^{i}g_{I }\left( t \right)\longleftrightarrow \left( -1 \right)^{i}\frac{\mathrm{d}^{i}}{\mathrm{d}s^{i}}\mathcal{L}_{I }\{g_{I }\left( t \right)\}\left( s \right)$. 
As far as the   Laplace transforms  $\mathcal{L}_{E_{w\mathrm{L}}}\left( s \right)$, $\mathcal{L}_{I_{\etv_{t}w\mathrm{L}}}\!\left(s \right)$, and $\mathcal{L}_{I_{\etv_{r}w\mathrm{L}}}\!\left(s \right)$ are concerned, these are provided by means of Lemma~\ref{LaplaceTransformGamma}. Finally, the  Laplace transform $\mathcal{L}_{I_{ \mathrm{L}}}\!\!\left(s \right)$ is obtained by Proposition~\ref{LaplaceTransform}. 
The conditional coverage probability $p_{c,\mathrm{N}} $, denoting the association with the NLOS BSs, is obtained after following the same approach with $p_{c,\mathrm{L}} $. In order to avoid any repetition, we omit the details.

\section{Proof of Proposition~\ref{LaplaceTransform}}\label{LaplaceTransformproof}
In order to derive the Laplace transform of the interference power $g_{jl}$ when the typical user is associated with a LOS BS, it is necessary to employ its PDF, which has been characterised in Appendix~\ref{SINRproof} as $\Gamma (K_{j},1)$. In other words, the PDF depends on the number of users, which is assumed that it is the same across all the cells in tier $j$. Note that we omit the index $\mathrm{L}$ from some variables for the sake of simplicity. Specifically, we have
\begin{align}
 &\!\!\mathcal{L}_{I_{ }}\!\!\left(  s_{w } \right)\!=\!\mathbb{E}_{I_{jl}}\!\!\left[  e^{-s_{w } {I_{jl}}}\right]\!=\!\mathbb{E}_{\Phi_{{B}_{j}},I_{Lj},I_{Nj}}\!\!\left[\!  e^{-s_{w }\!\! \sum{j \in \mathcal{W}}\left(I_{Lj}+I_{Nj} \right)}\!\right]\nn\\
  &\!\mathop = \prod_{j \in \mathcal{W}}\mathbb{E}_{\Phi_{B_{j}},I_{Lj},I_{Nj}}\!\left[ e^{-s_{w } \left(I_{Lj}+I_{Nj} \right)} \right] \nn\\
  &\!\mathop = \prod_{j \in \mathcal{W}}\mathbb{E}_{\Phi_{L_{j}},I_{Lj}}\!\left[ e^{-s_{w } I_{Lj} } \right]
  \mathbb{E}_{\Phi_{N_{j}},I_{Nj}}\!\left[ e^{-s_{w } {Nj} } \right],
  \label{laplace 21}
\end{align}
where $s_{w}=\frac{{T}_{w}\beta_{w}^{-1}}{\sigma_{\hat{\bh}_{w}}^{2}} R_{w}^{-\al}$, and~\eqref{laplace 21} results from the independence betweem the LOS and NLOS BSs in terms of location and  powers of the corresponding fading distributions.  Now, we focus on the term describing the LOS interfering links. We have
\begin{align}
&\!\!\!\!\mathbb{E}\!\left[ \!e^{-s_{w }\! I_{Lj} } \!\right]\!\!=\!\EE\!\left[\! e^{-s_{w }\! \sum_{l:X_{jl}\in \Phi_{\mathrm{L}_{j}}\cap\bar{B}_{j}\left( 0, x_{j} \right)}{\!G_{jl} C_{\mathrm{L}_{j}}R_{jl}^{-\al_{\mathrm{L}_{j}}}\! g_{jl}}} \!\right]\label{laplace 20}\!\!\\
&=e^{\left( -{2\pi\lambda_{w}}\sum_{k=1}^{4}b_{k}\int_{x}^{\infty} \left( 1-\mathcal{L}_{g_{jl}}\left(C_{\mathrm{L}_{j}} s_{w }  \left( \frac{y}{t} \right)^{-\al_{\mathrm{L}}} \right)   \right)p\left( t \right)\mathrm{d}t\right)}\label{laplace 4}\\
&=\prod_{k=1}^{4}e^{\!\left( -{2\pi\lambda_{w}}b_{k}\int_{x}^{\infty} \left( 1-\frac{1}{\left( 1+\frac{\bar{f}_{wk}C_{\mathrm{L}_{j}}s_{w }}{K_{j}}\left( \frac{y}{t} \right)^{-\al_{\mathrm{L}}} \right)^{K_{j}}}  \right)p\left( t \right)\mathrm{d}t\right)},\label{laplace 5}
\end{align}
where in~\eqref{laplace 20} we  have substituted $I_{Lj}$. Note that $I_{Lj}=\sum_{l:X_{jl}\in \Phi_{\mathrm{L}_{j}}\cap\bar{B}_{j}\left( 0, x_{j} \right)}{G_{jl} C_{\mathrm{L}_{j}}R_{jl}^{-\al_{\mathrm{L}_{j}}} h_{jl}}$ and $I_{Nj}=\sum_{l:X_{jl}\in \Phi_{\mathrm{N}_{j}}\cap\bar{B}_{j}\left( 0, \psi_{\mathrm{L}_{j}\left( x_{j} \right)} \right)}{G_{jl} C_{\mathrm{N}_{j}}R_{jl}^{-\al_{\mathrm{N}_{j}}} h_{jl}}$ are the interference powers from the LOS and NLOS BS of the $j$th tier, respectively. Next, we apply the property of the probability generating functional (PGFL) regarding the PPP~\cite{Chiu2013a}  to obtain~\eqref{laplace 4}. Moreover, we substitute the Laplace transform of the fading distribution
following a Gamma distribution with $K_{j}$ parameter, and $\bar{f}_{wk}=\frac{a_{wk}}{M_{rw}M_{tw}}$. Similarly, the other term concerning the NLOS interfering links is written as
\begin{align}
&\!\!\!\!\mathbb{E}\!\left[ \!e^{-s_{w }\! I_{Nj} } \!\right]\!\!=\!\EE\!\left[\! e^{-s_{w }\! \sum_{l:X_{jl}\in \Phi_{\mathrm{N}_{j}}\cap\bar{B}_{j}\left( 0, \psi_{j}(x_{j}) \right)}{\!G_{jl} C_{\mathrm{N}_{j}}R_{jl}^{-\al_{\mathrm{N}_{j}}}\! g_{jl}}} \!\right]\nn\!\!\\
&=e^{\left( -{2\pi\lambda}_{w}\sum_{k=1}^{4}b_{k}\int_{\psi_{j}(x_{j})}^{\infty} \left( 1-\mathcal{L}_{g_{jl}}\left(C_{\mathrm{N}_{j}} s_{w }  \left( \frac{y}{t} \right)^{-\al_{\mathrm{N}}} \right)   \right)p\left( t \right)\mathrm{d}t\right)}\nn\\
&\!=\!\prod_{k=1}^{4}\!\!e^{\!\!\left(\!\! -{2\pi\lambda_{w}}b_{k}\int_{\psi_{j}(x_{j})}^{\infty}\! \left(\!\! 1-\frac{1}{\left( 1+\frac{\bar{f}_{wk} C_{\mathrm{N}_{j}}s_{w }}{K_{j}}\left( \frac{y}{t} \right)^{-\al_{\mathrm{N}}} \right)^{K_{j}}}  \!\!\right)\!\left( 1-p\left( t \right) \right)\mathrm{d}t\!\!\right)}\!\!.\!\label{laplace 51}
\end{align}
Substituting~\eqref{laplace 5} and~\eqref{laplace 51} into~\eqref{laplace 21}, we conclude the proof.

 \end{appendices}

\bibliographystyle{IEEEtran}

\bibliography{mybib}

\end{document}

%% file: def.tex
\usepackage{calc}

\usepackage{xcolor,multirow,gensymb}
  \usepackage{cite}
\usepackage{graphicx,subfigure,epsfig}
\usepackage{psfrag}
\usepackage{amsmath,amssymb}
\usepackage{color}
\usepackage{array}
\usepackage{nicefrac}
\usepackage{amsmath,amssymb}
\usepackage{xifthen}
\usepackage{mathtools}
\usepackage{enumerate}
\usepackage{microtype}
\usepackage{xspace}
\usepackage{bm}
\usepackage[T1]{fontenc}
\usepackage{fancyhdr}
\usepackage{lastpage}
\usepackage{rotating}
\usepackage{tabulary}

\newcommand{\vect}[1]{{\lowercase{\mbs{#1}}}}

\newcommand{\mbs}[1]{\bm{#1}}
\newcommand{\mat}[1]{{\uppercase{\mbs{#1}}}}
\newcommand{\Id}{\mat{\mathrm{I}}}

\newcommand{\T}{{\scriptscriptstyle\mathsf{T}}}
\renewcommand{\H}{{\scriptscriptstyle\mathsf{H}}}

\renewcommand{\Re}[1][]{\ifthenelse{\isempty{#1}}{\operatorname{Re}}{\operatorname{Re}\left(#1\right)}}
\renewcommand{\Im}[1][]{\ifthenelse{\isempty{#1}}{\operatorname{Im}}{\operatorname{Im}\left(#1\right)}}
\usepackage{float}

\newfloat{longequation}{b}{ext}

\newcommand{\bv}{\vect{b}}

\newcommand{\etv}{\vect{\eta}}

\def\bH{{\mathbf{H}}}
\def\bI{{\mathbf{I}}}

\def\bQ{{\mathbf{Q}}}

\def\bV{{\mathbf{V}}}

\newcommand{\cC}{{\cal C}}

\newcommand{\cN}{{\cal N}}

\def\bb{{\mathbf{b}}}

\def\bd{{\mathbf{d}}}
\def\bee{{\mathbf{e}}}

\def\bg{{\mathbf{g}}}
\def\bh{{\mathbf{h}}}

\def\bs{{\mathbf{s}}}

\def\bv{{\mathbf{v}}}

\def\b0{{\mathbf{0}}}

\def\bbC{{\mathbb{C}}}

\newcommand{\EE}{\mathbb{E}}

\newcommand{\CN}[1][]{\ifthenelse{\isempty{#1}}{\mathcal{N}_{\mathbb{C}}}{\mathcal{N}_{\mathbb{C}}\left(#1\right)}}

\renewcommand{\P}[1][]{\ifthenelse{\isempty{#1}}{\mathbb{P}}{\mathbb{P}\left(#1\right)}}
\newcommand{\E}[1][]{\ifthenelse{\isempty{#1}}{\mathbb{E}}{\mathbb{E}\left(#1\right)}}
\renewcommand{\det}[1][]{\ifthenelse{\isempty{#1}}{\text{det}}{\text{det}\left(#1\right)}}
\newcommand{\trace}[1][]{\ifthenelse{\isempty{#1}}{\text{tr}}{\text{tr}\left(#1\right)}}
\newcommand{\rank}[1][]{\ifthenelse{\isempty{#1}}{\text{rank}}{\text{rank}\left(#1\right)}}
\newcommand{\diag}[1][]{\ifthenelse{\isempty{#1}}{\text{diag}}{\text{diag}\left(#1\right)}}
\def\nn{\nonumber}

\IEEEoverridecommandlockouts

\newtheorem{proposition}{Proposition}
\newtheorem{remark}{Remark}
\newtheorem{definition}{Definition}
\newtheorem{theorem}{Theorem}
\newtheorem{lemma}{Lemma}
\newtheorem{assumption}{Assumption}

\newcommand{\tr}{\mathop{\mathrm{tr}}\nolimits}

\newcommand{\al}{\alpha}

\newcounter{enumi_saved}
\usepackage{answers}
\Newassociation{solution}{Solution}{solutionfile}

\AtBeginDocument{\Opensolutionfile{solutionfile}[\jobname]}
\AtEndDocument{\Closesolutionfile{solutionfile}\clearpage
}
